\documentclass[a4paper,11pt]{article}
\pdfoutput=1
\usepackage{jcappub}
\usepackage[T1]{fontenc}
\usepackage[utf8]{inputenc}
\usepackage{amsthm}
\usepackage{booktabs}
\usepackage{aas_macros}
\hypersetup{pdftitle={An E-value Reanalysis of DESI's Dynamical Dark Energy Signal},
            pdfauthor={Jinyoung Kim, David F. Mota, Andrius Tamosiunas}}

\renewcommand{\aap}{Astron. Astrophys.}
\newtheorem{definition}{Definition}
\newtheorem{theorem}{Theorem}
\newtheorem{corollary}{Corollary}
\newtheorem{proposition}{Proposition}
\newtheorem{remark}{Remark}
\newcommand{\E}{\mathbb{E}}
\newcommand{\Prob}{\mathbb{P}}
\newcommand{\LCDM}{$\Lambda$CDM}

\title{An E-value Reanalysis of DESI's Dynamical Dark Energy Signal}
\author[a,d,1]{Jinyoung Kim,\note{Corresponding author. This work is independent research; the views expressed do not represent those of the author's employer.}}
\author[b]{David F. Mota,}
\author[b,c]{and Andrius Tamosiunas}
\affiliation[a]{Department of Computer Science, Stanford University, 353 Jane Stanford Way, Stanford, CA 94305, United States}
\affiliation[b]{Institute of Theoretical Astrophysics, Universitetet i Oslo, 0315 Oslo, Norway}
\affiliation[c]{CERCA/ISO, Department of Physics, Case Western Reserve University, Cleveland, OH 44106, United States}
\affiliation[d]{Two Sigma Investments, 100 Avenue of the Americas, New York, NY 10013, United States}
\emailAdd{jinkim2@stanford.edu}
\emailAdd{d.f.mota@astro.uio.no}
\emailAdd{andrius.tamosiunas@astro.uio.no}
\abstract{Cosmological surveys such as DESI release data in stages, and each release reprocesses earlier observations and supersedes its predecessor. Testing the same hypothesis at every release inflates the false-positive rate, and only the collaboration can estimate the cross-release covariance that an exact combination needs. We show how to combine evidence across such releases using public products alone. From each release we compute an e-value, a non-negative statistic with expectation at most one under the null, and we average these e-values with weights fixed in advance. Provided each release's own likelihood is correct, this average controls the false-positive rate at every release under arbitrary dependence between releases. For DESI BAO at a fixed Planck background, DR2 alone gives an e-value of $33.97$ against $\Lambda$CDM ($p \le 0.0295$) under our reference test specification, a flat prior on $(w_0, w_a)$ over a fixed box. The equal-weight average of DR1 and DR2 is $17.51$ ($p \le 0.0571$), below the rejection threshold of $20$. We did not fix these weights before DR2 was published, and the weights we now fix ahead of DR3 give $p \le 0.0874$ at DR2. Without a model of the cross-release covariance, an exact combination that counts each observation once can take any value. Under explicit models it ranges from below the threshold to nearly double the DR2 value, depending on the structure assumed. The evidence is concentrated in the LRG2 bin, which also changed most between releases. Without it the DR2 e-value falls to $0.49$, and a look-elsewhere correction over the seven redshift bins leaves no rejection. DR2 does not reject under wider priors either. The signal is therefore fragile and specification-dependent. We recommend that surveys report per-release e-values with a stated test specification and a combination rule fixed in advance, and publish the cross-release covariance.}

\arxivnumber{2607.28918}
\begin{document}

\setcounter{tocdepth}{2}
\maketitle
\flushbottom

\section{Introduction}
\label{sec:intro}
The standard cosmological model, $\Lambda$ cold dark matter (\LCDM), describes a spatially flat universe in which a cosmological constant $\Lambda$ with equation of state $w = -1$ drives the late-time acceleration. Its six free parameters have been enough to fit cosmic microwave background (CMB) anisotropies, type Ia supernovae, weak lensing, galaxy clustering and baryon acoustic oscillations (BAO) \cite{Planck2020, Riess1998, Perlmutter1999, Kilbinger2015, Tegmark2004, Eisenstein2005}. Yet the same data sets have begun to expose tensions. The most prominent are the $\sim 5\sigma$ discrepancy between local and CMB-inferred values of the Hubble constant \cite{DiValentino2021} and the $S_8$ tension between the primary CMB and weak-lensing measurements of late-time clustering \cite{Poulin2023}.

The Dark Energy Spectroscopic Instrument (DESI) has produced the most precise BAO measurements to date. The joint fit of its second data release (DR2) with CMB and type~Ia supernova (SN) data prefers a time-evolving dark energy equation of state ($w_0w_a$CDM) over the cosmological constant at $\sim 3$--$4\sigma$ \citep{DESI2025}. The preference depends on which supernova compilation (Pantheon+, Union3, DES-Y5) enters the fit, and it runs from $\sim 2.8\sigma$ to $4.2\sigma$.

Several analyses have argued that this significance is fragile \cite{Cortes2024, WangMota2025, Colgain2025, Efstathiou2025, Ormondroyd2026, Gialamas2025, Li:2026asg}, and nested-sampling model comparison favours \LCDM\ ($\ln\mathcal{B} = -0.57 \pm 0.26$)~\citep{OngBayesian2025}. One reason is that the signal is not spread across the survey. One analysis flagged a statistical fluctuation in LRG1 immediately after the first release (DR1)~\citep{ColgainDESI2024}. Removing the LRG1 and LRG2 bins largely removes the DR1 preference~\citep{Liu2024LRG}, and it still does when CMB and SN data enter the fit~\citep{Wang2024DR1}. In DR2, LRG2 drives most of the $w_0 > -1$ preference~\citep{ColgainEvolved2025, KarimSpeagleHlozek2026}, and removing LRG1 and LRG2 together restores \LCDM~\citep{Chaudhary2025}. The modelling choices matter as well. Full-shape modelling of the same data shows no such preference~\citep{ColgainAnalysis2025}. The Chevallier--Polarski--Linder (CPL) form $w(a) = w_0 + w_a(1-a)$~\citep{Chevallier2001, Linder2003} is one expansion among many, and the preference shifts with the choice~\citep{Nesseris2025, Artola2026, Alam2026, Lee2026}.

A quoted significance also changes meaning when analysts ask the same question again each time more data arrive. This happened with BICEP2's primordial $B$-mode claim~\citep{BICEP2014, BKP2015} and with the EDGES 21-cm absorption feature~\citep{Bowman2018, Singh2022}. Each DESI release reports a $\sigma$ value and so reruns the same hypothesis test. Repeating the test raises the false-positive rate unless a rule fixed in advance combines the looks. The standard remedy is the Bonferroni correction, which holds the overall false-positive rate at $\alpha$ by tightening the per-look threshold in proportion to the number of looks. The per-look $\sigma$ it demands therefore keeps rising as releases accumulate. At a two-sided overall $\alpha = 0.05$ the threshold is $1.96\sigma$ for a single look, $2.24\sigma$ at two looks and $2.58\sigma$ at five.\footnote{Bonferroni is worst-case by construction: the union bound $n\alpha$, here $0.10$ for two looks at $\alpha = 0.05$, holds whatever the dependence between the looks, and independent looks would give the slightly smaller $1-(1-\alpha)^n$.}

We take a different route and apply e-values~\citep{Vovk2021, Shafer2021, Ramdas2023, GrunwaldBeyondNeymanPearson, GrunwaldSafeTesting2024, Howard2021}, non-negative statistics whose expectation under the null is at most one, to the DESI BAO data. E-values have appeared before in galaxy-velocity modelling, as a statistical illustration~\citep{Dombowsky2026}. Our per-release statistic, defined in Section~\ref{sec:evalues}, is the mixture e-value
\begin{equation}
M_t = \frac{\int L(\mathcal D_{t}\mid\theta)\,\pi(\theta)\,d\theta}{L(\mathcal D_{t}\mid H_0)},
\label{eq:Mt}
\end{equation}
where $\mathcal D_{t}$ is the BAO data vector published with DESI release $t$ and $\pi$ is a fixed test specification on $\theta = (w_0, w_a)$. Each $M_t$ has expectation one under $H_0$ at its own release.

DESI's releases supersede one another instead of accumulating, so we cannot multiply the per-release values as if they came from independent data. DR2 contains the DR1 observations, reprocessed with changed sample definitions, fitting choices and systematic-error treatment~\citep{Andrade2025}. Other surveys do the same. The DES-Dovekie supernova reanalysis supersedes the five-year DES-SN5YR distances~\citep{DESDovekie2026}, and the final KiDS release reanalyses the whole survey in place of KiDS-1000~\citep{KiDSLegacy2025}. An exact combination across releases needs the cross-release covariance. Only the collaboration can estimate that covariance, with mocks run through each release's own pipeline~\citep{ForeroSanchez2026}. One recent outside analysis shows the consequence. It combined DESI's DR1 full-shape and DR2 BAO results, and it had to treat them as independent because the cross-correlation ``is not yet publicly available''~\citep{Clini2026}.

We instead combine the per-release values by a weighted average with weights fixed in advance. Such an average is an e-value whatever the relation between the releases, provided each release's own likelihood is correct, so a single threshold holds at every release at once (Section~\ref{sec:eprocess}). We adopt an illustrative level $\alpha = 0.05$ throughout, so an e-value of $1/\alpha = 20$ or more rejects. DR2 alone reaches $20$ under our reference specification, but the equal-weight average of DR1 and DR2 is $17.51$ and does not (Section~\ref{sec:results}). DR3 joins a weighted average of the same kind, with weights we state now (Section~\ref{sec:forecast}), so we fix its share of the combination before its data exist, as in a weighted Bonferroni correction.

We make four contributions. First, we pose the problem of repeated significance testing across survey releases that supersede one another, for an analyst who has only the public products. Collaborations can combine releases internally with mocks~\citep{ForeroSanchez2026, ForeroSanchez2026b}, and the nearest statistical work, on meta-analyses updated as new trials arrive~\citep{terSchure2025}, does not reanalyse old data. Second, we give the rule that solves the problem: average the per-release e-values with weights fixed in advance. That average stays valid under any relation between releases, and equal weights are essentially optimal among symmetric rules. Both facts are known results on merging e-values~\citep{Vovk2021, RamdasWang2025}, and we apply them here to a new problem. Third, we show that the alternative, an exact combination that scores each observation once, says nothing without a model of how the releases are related and depends on the model's structure with one. Over all joint Gaussian laws consistent with the published covariances, the statistic can be driven to zero or to any large value. One noise fraction common to all bins gives $64.7$, the correlations DESI uses for its own DR1--DR2 consistency check imply $28.9$ to $33.7$, and one fraction per bin gives $16.2$, below the threshold (Section~\ref{sec:eprocess}, Appendix~\ref{app:hierarchical}). Fourth, we localise the DESI signal with calibrated e-values. The rejection rests on LRG2, and a look-elsewhere correction over the seven bins leaves no rejection. That bin is also the one that changed most between DR1 and DR2, and a smooth equation of state rarely puts this much of the evidence there (Section~\ref{sec:loo}).

Section~\ref{sec:evalues} defines the e-values and the test specifications, and Section~\ref{sec:methods} sets out the DESI likelihood and the fixed background. Section~\ref{sec:results} then reports the per-release values, combines them across releases and stress-tests them, before Section~\ref{sec:discussion} concludes. Appendix~\ref{app:proofs} collects the propositions and proofs, and Appendix~\ref{app:methods-notes} gives the supporting methodology and sensitivity checks.

\section{E-values for model selection across releases}
\label{sec:evalues}
Let $H_0$ denote the null (\LCDM, $w = -1$) and $H_1$ the alternative ($w_0w_a$CDM, $w(a) = w_0 + w_a(1-a)$). An e-value is a non-negative random variable $E$ whose expectation under the null is at most one, $\E[E \mid H_0] \leq 1$. Large observed values are therefore unlikely under $H_0$, and they count as evidence against it. One may read the letter ``e'' as the \emph{expectation} in the definition or as the \emph{evidence} the statistic quantifies~\cite{RamdasWang2025}. Markov's inequality turns this into an error bound. It gives $\Prob(E \geq 1/\alpha \mid H_0) \leq \alpha$, so rejecting $H_0$ at $E \geq 1/\alpha$ is Type-I-controlled at $\alpha$. This bound needs no Wilks or asymptotic approximation, but it still assumes the Gaussian BAO likelihood with the published covariance, which is the same model the $\chi^2$ analysis uses. An e-value quantifies relative evidence, as a Bayes factor does, and it controls Type-I error, as a $p$-value does.

An analyst may look after DR1, again after DR2 and after later releases, and the threshold stays valid, provided a rule fixed in advance combines the releases (Section~\ref{sec:eprocess}). E-values from independent data multiply. Averaging, by contrast, needs no independence at all, because a weighted average of e-values, with non-negative weights fixed in advance of the data and summing to at most one, is again an e-value whatever the dependence between them~\citep{Vovk2021, RamdasWang2025}. We use this average to combine the DESI releases.

We use the \emph{mixture e-value}, which compares how well a specified class of $w_0w_a$CDM models describes a release's BAO data, on average, relative to \LCDM. To build it, we fix a probability measure $\pi$ on the alternative parameters $(w_0, w_a)$ that does not depend on the data (we state below when we fixed ours) and average the likelihood ratio over it,
\begin{equation}
E_{\mathrm{mix}} = \int \frac{L(\mathrm{data} \mid \theta)}{L(\mathrm{data} \mid H_0)} \, \pi(\theta) \, d\theta,
\label{eq:Emix}
\end{equation}
which is exactly the Bayes factor with prior $\pi$. The mixture is a valid e-value because each integrand $L(\mathrm{data}\mid\theta)/L(\mathrm{data}\mid H_0)$ is a simple-versus-simple likelihood ratio with null expectation one, and averaging over a proper $\pi$ preserves that expectation, so $\E[E_{\mathrm{mix}} \mid H_0] = 1$ (Proposition~\ref{prop:mixture} for a grid; the general case follows in the same way by Fubini's theorem). By contrast, a maximised likelihood ratio, $e^{\Delta\chi^2/2}$, is not an e-value (Appendix~\ref{app:proofs}).

We call $\pi$ the ``prior'' for consistency with the Bayes-factor literature. Here, though, it works as a frequentist test specification. It selects the alternative against which we design the test to have power, and it expresses no Bayesian belief about $(w_0, w_a)$.

Our default analysis, and our reference specification, takes $\pi$ to be a flat prior on a $30 \times 30$ grid in $[-1.5, -0.5] \times [-2.0, 1.0]$. We fixed this Default box before we computed any e-value, after DR2 was public, and sized it to contain the published joint SN+CMB maximum a posteriori (MAP) points of Pantheon+, Union3 and DES-Y5~\cite{BroutPantheon2022,DESY5_2024,RubinUnion3_2023}. The box also contains the physically plausible thawing and freezing regions, and it is narrower than the diffuse reference box of Ong et al.~\cite{OngBayesian2025}.

We also run other test specifications (Section~\ref{sec:results}). Box-width variants (Narrow, Wide and the Ong et al.\ box) probe how much prior mass sits away from the data's preferred direction. The REGROW priors, defined in the next paragraph, instead swap a box boundary for a pre-specified minimum effect size, and the thawing and freezing priors replace the reference boxes with dark-energy-model predictions.

GROW (growth-rate optimal) denotes the prior maximising the worst-case expected log-e-value over a restricted alternative~\cite{GrunwaldSafeTesting2024}. We take that alternative to be $\{\theta : d_F(\theta, \theta_0) \geq \delta\}$ with $\theta_0$ the \LCDM\ point and $d_F$ the Fisher distance, measured in Fisher-$\sigma$, where one Fisher-$\sigma$ is one standard deviation in the metric set by the DR2 Fisher information $F$. The prefix ``RE-'' marks the GROW prior \emph{re}stricted to that minimum-effect-size set. In the local-Gaussian regime, this restricted prior concentrates its mass on the Fisher-distance-$\delta$ ellipse around \LCDM. We borrowed the name, but our construction differs from the relative-growth criterion that Gr\"unwald et al.~\cite{GrunwaldSafeTesting2024} call REGROW. Table~\ref{tab:summary} and Section~\ref{sec:prior-sens} give the sensitivity analysis over the other specifications.

\section{Data and likelihood construction}
\label{sec:methods}

We use the official DESI DR2 BAO measurements from the CobayaSampler repository.\footnote{\url{https://github.com/CobayaSampler/bao_data}.} The data set contains 13 measurements across seven redshift bins (BGS at $z_{\text{eff}} = 0.295$, LRG1, LRG2, LRG3+ELG1, ELG2, QSO, Ly$\alpha$ up to $z_{\text{eff}} = 2.330$), and their means match Table~IV of Ref.~\citep{DESI2025}. DESI labels the bins by tracer: the bright galaxy sample (BGS), luminous red galaxies (LRG), emission-line galaxies (ELG), quasars (QSO) and the Ly$\alpha$ forest. Each bin reports the transverse and radial BAO scales in units of the sound horizon at the drag epoch $r_d$, namely the comoving angular diameter distance $D_M(z)$ and the Hubble distance $D_H(z) = c/H(z)$. A bin that cannot separate the two reports the isotropic combination $D_V(z) = [z D_M^2 D_H]^{1/3}$ instead. Figure~\ref{fig:bao_data} shows the data against the \LCDM\ prediction and against the $w_0w_a$CDM prediction at the DR2 BAO maximum-likelihood estimate (MLE) for the Planck-fixed background, $(w_0, w_a) = (-0.856, -0.430)$.

We take the covariance from the released Gaussian likelihood, and the choice matters at the factor-of-two level. The $13\times 13$ matrix is block-diagonal by bin. Within each bin, $D_M/r_d$ and $D_H/r_d$ are anti-correlated with $\rho \approx -0.35$ to $-0.49$. The released errors exceed the posterior summaries printed in Table~IV, by up to $6.4\%$ for LRG3$+$ELG1, whose $D_M$--$D_H$ correlation is also weaker ($-0.35$ against $-0.42$), and by about $2\%$ or less elsewhere. The DR2 statistic of Section~\ref{sec:results} would be $68.4$ rather than $33.97$ if we used the tabulated errors and correlations instead ($52.9$ if we replaced only the LRG3$+$ELG1 block). We use the released likelihood because it is the product distributed for likelihood analyses and because DESI used it for its own DR2 cosmology chains.\footnote{DESI's data-portal README for the DR2 BAO cosmology results, \url{https://data.desi.lbl.gov/public/papers/y3/bao-cosmo-params/README.html}, states that the likelihoods used to produce those results are the files distributed in the CobayaSampler repository.}

The same repository also supplies DR1 as the corresponding 12-measurement vector and covariance~\citep{DESI2024}, which agree with Table~1 of that paper to the printed precision. DR1 reports only $D_V/r_d$ for the QSO bin. This difference matters for the release-to-release matching of Appendix~\ref{app:hierarchical}.

\begin{figure}
    \centering
    \includegraphics[width=\textwidth]{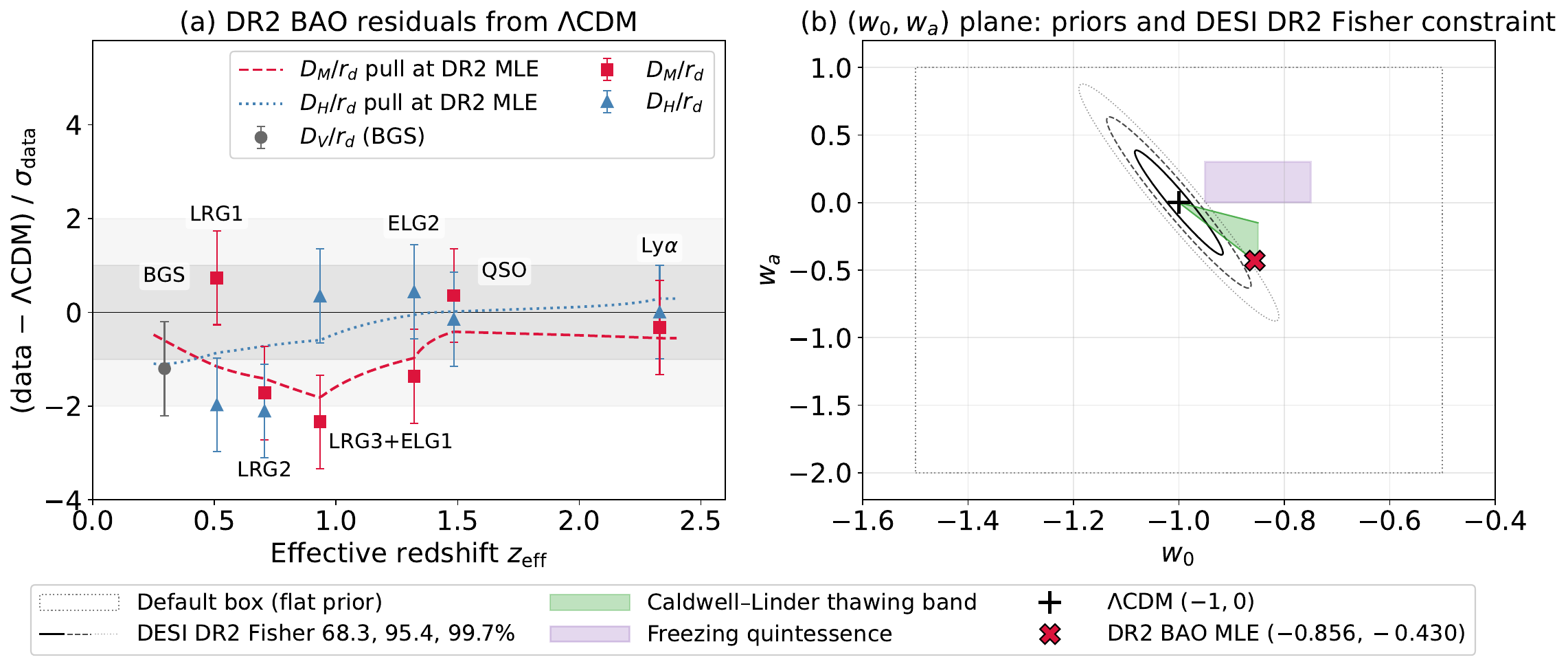}
    \caption{(a) DESI DR2 BAO residuals from the Planck 2018 \LCDM\ prediction $(h, \Omega_m, r_d) = (0.6766, 0.3111, 147.05~\mathrm{Mpc})$, in measurement-$\sigma$ units; the grey bands mark $\pm 1\sigma$ and $\pm 2\sigma$. The dashed ($D_M/r_d$) and dotted ($D_H/r_d$) curves show the $w_0w_a$CDM prediction at the DR2 BAO MLE $(w_0, w_a) = (-0.856, -0.430)$ as its offset from \LCDM\ in the same units, with $\sigma$ interpolated between bins; they are not fits to the points. Here $D_M$ is the comoving angular diameter distance and $D_H = c/H(z)$ the Hubble distance, each in units of the drag-epoch sound horizon $r_d$. (b) The $(w_0, w_a)$ plane: DR2 Fisher ellipses around \LCDM\ (two-parameter $68.3$, $95.4$ and $99.7\%$ contours), the BAO MLE, the Caldwell--Linder thawing-quintessence band (green), the freezing region (purple) and the Default box (dotted).}
    \label{fig:bao_data}
\end{figure}

We fix the background cosmology to the Planck 2018 values $h = 0.6766$, $\Omega_m = 0.3111$, $r_d = 147.05$~Mpc~\cite{Planck2020}, so that only $w_0$ and $w_a$ vary and the effective number of additional parameters is $k = 2$. Distances then reduce to one-dimensional integrals of a closed-form expansion rate. BAO-only analyses commonly adopt this simplification. The triple does, however, mix Planck columns, since $(h, \Omega_m)$ come from the +lensing+BAO chain and $r_d$ from the CMB-only chains. Appendix~\ref{app:fixed-planck} therefore reruns the analysis on self-consistent columns.

The mixture $M_t$ uses the full likelihood, so it does not depend on how we group the 13 measurements. The leave-one-out and per-bin-independent product constructions (Section~\ref{sec:loo}) do depend on it, because they rely on independence between the units they fold over, and the two points within a bin are anti-correlated. We therefore build them at the bin level rather than the point level, which matches the block-diagonal covariance.

For the cross-dataset analysis of Appendix~\ref{app:cross}, we take each SN$+$CMB compilation's published best-fit and full $(w_0, w_a)$ posterior~\cite{BroutPantheon2022, DESY5_2024, RubinUnion3_2023} as a pre-specified alternative against DESI BAO.

\section{Results}
\label{sec:results}

DR2 rejects as a single look under the Default prior. The Default box is a moderately concentrated alternative, and under it the per-release mixture e-value is $M_{\mathrm{DR1}} = 1.05$ at DR1 and $M_{\mathrm{DR2}} = 33.97$ at DR2. We computed both values at one fixed Planck background (Section~\ref{sec:methods}), and the result depends on that choice. Across the four Planck 2018 backgrounds of Appendix~\ref{app:fixed-planck}, $M_{\mathrm{DR2}}$ runs from $8.70$ to $1417$, and the worst case over them, $8.70$ ($p \le 0.115$), does not reject. Section~\ref{sec:eprocess} combines the two releases, and Sections~\ref{sec:loo} and~\ref{sec:prior-sens} test how well the rejection holds up.

Table~\ref{tab:summary} reports the per-release values under seven of the test specifications introduced in Section~\ref{sec:evalues}. We treat the thawing and freezing priors in Section~\ref{sec:prior-sens}.

\begin{table}[t]
\centering
\small
\begin{tabular}{@{}lcccc@{}}
\toprule
Prior on $(w_0,w_a)$ & $M_{\mathrm{DR1}}$ & $M_{\mathrm{DR2}}$ & Markov bound on $p$ & Reject \LCDM \\
\midrule
\multicolumn{5}{l}{\textit{Flat ``box'' priors:}}\\
\quad Narrow $[-1.2,-0.8]\times[-1.0,0.5]$ & 2.60 & 144 & 0.0070 & $\checkmark$ \\
\quad Default $[-1.5,-0.5]\times[-2.0,1.0]$ & 1.05 & 33.97 & 0.0295 & $\checkmark$ \\
\quad Wide $[-2.0,0.0]\times[-3.0,2.0]$ & 0.320 & 10.7 & 0.094 &  \\
\quad Ong et al.~\cite{OngBayesian2025} $[-3.0,1.0]\times[-3.0,2.0]$ & 0.160 & 4.81 & 0.208 &  \\
\midrule
\multicolumn{5}{l}{\textit{REGROW (Fisher-$\delta$ ellipse around \LCDM):}}\\
\quad $\delta = 1$ Fisher-$\sigma$ & 1.50 & 7.58 & 0.132 &  \\
\quad $\delta = 2$ Fisher-$\sigma$ & 3.02 & 72.5 & 0.014 & $\checkmark$ \\
\quad $\delta = 3$ Fisher-$\sigma$ & 5.25 & 297 & 0.0034 & $\checkmark$ \\
\bottomrule
\end{tabular}
\caption{Per-release mixture e-values $M_t$ for DESI DR1 and DR2 BAO under seven test specifications on $(w_0, w_a)$, at the Planck-fixed background. A tick marks $M_{\mathrm{DR2}} \geq 20$, the rejection threshold, for DR2 analysed as a single look (Section~\ref{sec:eprocess}). The equal-weight average with DR1 reaches $20$ only for Narrow and REGROW $\delta \geq 2$. The Markov bound is $p \leq 1/M_{\mathrm{DR2}}$, rounded up.}
\label{tab:summary}
\end{table}

\subsection{Combining the evidence across DR1 and DR2}
\label{sec:eprocess}

Each $M_t$ is a valid e-value at its own release ($\E[M_t \mid H_0] = 1$ for proper $\pi$, by Fubini's theorem), but a guarantee across releases needs more than that. DR2 contains the DR1 observations, reprocessed with the new ones in one homogeneous processing run~\citep{DESI2025}, and Appendix~\ref{app:hierarchical} lists what changed~\citep{Andrade2025}. The published means and covariances therefore do not determine the joint distribution of the two releases, so we cannot multiply the per-release values.

Instead we report the combined evidence in two statements that need no model of the relation between releases, and in one that does (Figure~\ref{fig:ladder}).

\begin{figure}[t]
\centering
\includegraphics[width=0.92\textwidth]{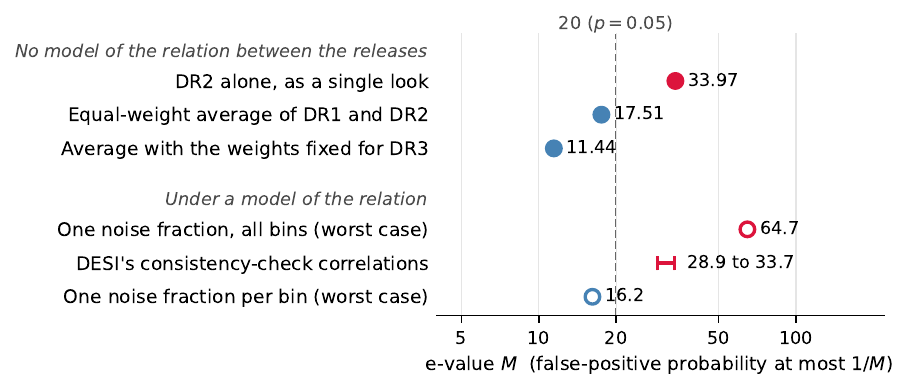}
\caption{What each statistic from DR1 and DR2 delivers. Each row is one statistic of Section~\ref{sec:eprocess}, DR2 alone or a combination of the two releases, placed at its e-value; an e-value of $20$ bounds the false-positive probability at $0.05$ (dashed). Filled markers need no model of how the releases are related; hollow markers hold only under the stated cross-release model. Red rejects, blue does not.}
\label{fig:ladder}
\end{figure}

\emph{(i) One release.} DR2 analysed as a single look gives $M_{\mathrm{DR2}} = 33.97$, $p \le 0.0295$. The bound is valid for an analyst who tests DR2 alone, but it disregards the DR1 look.

\emph{(ii) The weighted average.} The equal-weight average $(M_{\mathrm{DR1}} + M_{\mathrm{DR2}})/2 = 17.51$ is an e-value whatever the relation between the releases, including arbitrary reprocessing, and it does not reach $20$ ($p \le 0.0571$). It is valid because the expectation of a sum is the sum of the expectations under any dependence (Proposition~\ref{prop:weighted}). Equal weights are the natural symmetric choice, and they are essentially optimal among symmetric rules that stay valid for every relation between the releases (Remark~\ref{rem:vovk-wang}, after~\citep{Vovk2021}). They do not depend on the observed values, but we did not fix them before DESI published DR2. We fix the weights for DR3 now, in proportion to each release's observing time (Section~\ref{sec:forecast}), and those weights give $11.44$ at DR2 ($p \le 0.0874$).\footnote{Other conventions give other bounds: $0.0435$ for weights $(1/3, 2/3)$ and $0.0389$ for $(1/4, 3/4)$, roughly in proportion to the two releases' precision; $0.0857$ for $1/3$ per release over three releases; any weight of $0.576$ or more on DR2 reaches $20$. Weight zero on DR1 and one half on DR2, which needs no assumption about DR1, gives $M_{\mathrm{DR2}}/2 = 16.99$, $p \le 0.0589$.}

The average gives one threshold that holds at every release under any dependence, but it does not accumulate evidence as a product would. Its weights sum to at most one, so it never exceeds the largest per-release value. At the same time the average is at least each per-release value times its weight, so it rejects whenever a release would under a weighted Bonferroni rule at its share, and sometimes when none does.

The average is valid under one assumption, that each release's own published likelihood is correct under \LCDM. DESI's quoted constraints, including those from the superseded DR1, rest on the same assumption.

\emph{(iii) Combination under a model of the relation.} A combined statistic that scores DR1 once and then only the information DR2 adds needs the cross-release covariance. The published products do not identify that covariance, so without a model of it the statistic says nothing. Over all joint Gaussian laws consistent with the two published covariances, a rank-one cross-covariance inside a single redshift bin can drive the statistic to zero or to arbitrarily large values (Appendix~\ref{app:hierarchical}).

The answer under a model depends on the model's structure more than on any parameter in it. One family lets DR2 inherit one fraction of the DR1 noise common to all bins, and the worst case over the allowed range of that fraction is $64.7$ ($p \le 0.0155$). A second family starts from the bin-by-bin correlations that DESI uses for its own DR1--DR2 consistency check~\citep{DESI2025}, and completing them to a cross-release covariance in three ways gives $28.9$ to $33.7$ ($p \le 0.035$). A third lets each bin take its own fraction, up to the value for strictly nested data, and there the worst case is $16.2$, below the threshold ($p \le 0.0619$). A model-based rejection therefore depends on the dimension of the family as well as on its form. Appendix~\ref{app:hierarchical} defines these families and gives the worst case of each. Only mocks run through both release pipelines can say which structure holds~\citep{ForeroSanchez2026}.

Each $p$ above is a Markov bound, $p \le 1/M$, rather than an estimate of a tail probability. For comparison, the Monte Carlo tail of $M_{\mathrm{DR2}}$ under the published likelihood is $2.25 \times 10^{-4}$ (Appendix~\ref{app:sigma}). A Bonferroni correction of that tail over the two looks, which is also valid whatever the relation between releases, gives $4.5 \times 10^{-4}$ (Monte Carlo, $95\%$ interval $3.2$ to $5.8 \times 10^{-4}$). We report the e-value nevertheless, because it is a calibrated measure of evidence against a stated alternative that stays valid across looks under a rule fixed in advance (Proposition~\ref{prop:weighted}), and because it decomposes by bin (Section~\ref{sec:loo}) and by specification (Section~\ref{sec:prior-sens}).

The rest of the paper quotes the DR2 value, $M_{\mathrm{DR2}} = 33.97$, which needs no cross-release assumption. Three checks put it in context. The corresponding Wilks statistic, BAO-only at fixed Planck, is $\Delta\chi^2 = 16.86$ on $\chi^2(2)$ ($\sigma = 3.70$ two-sided, $p = 2.2 \times 10^{-4}$), and it gives the same DR2 rejection while using neither the supernova nor the CMB data behind DESI's published figures. The rejection does not need the mixture prior either. Universal-inference e-values \citep{Wasserman2020} use no prior but fit the MLE on random subsets of the DR2 bins and score the held-out bins, and they give $E_{\mathrm{UI}} \approx 78$ (average over $400$ splits). The third check cuts the other way. Adding the compressed Planck CMB likelihood lowers the Default mixture e-value to $E^{\mathrm{BAO+CMB}}_{\mathrm{mix,Default}} = 2.19$, which does not reject (Appendix~\ref{app:joint-cmb}). This is one of several ways in which the rejection fails to hold up.

\subsection[Localising the evidence: leave-one-out and per-bin-independent decompositions]{Localising the evidence:\\ leave-one-out and per-bin-independent decompositions}
\label{sec:loo}

\begin{table}[t]
\centering
\small
\begin{tabular}{@{}lcccc@{}}
\toprule
Bin & $z_{\mathrm{eff}}$ & Fitted $(w_0, w_a)$ on six remaining bins & $E_k^{\mathrm{LOO}}$ & $M_k^{\mathrm{REGROW},\delta=2}$ \\
\midrule
BGS         & 0.295 & $(-0.845, -0.476)$ & 1.89 & 1.06 \\
LRG1        & 0.510 & $(-0.816, -0.596)$ & 0.71 & 0.96 \\
\textbf{LRG2}        & \textbf{0.706} & $\mathbf{(-0.881, -0.369)}$ & \textbf{55.98} & \textbf{4.25} \\
LRG3+ELG1   & 0.934 & $(-0.871, -0.332)$ & 8.73 & 1.84 \\
ELG2        & 1.321 & $(-0.865, -0.382)$ & 2.14 & 0.87 \\
QSO         & 1.484 & $(-0.857, -0.416)$ & 0.75 & 0.94 \\
Ly$\alpha$  & 2.330 & $(-0.850, -0.460)$ & 1.00 & 0.68 \\
\midrule
\multicolumn{3}{@{}l}{\textbf{LOO average} $\bar{E}^{\mathrm{LOO}} = \tfrac{1}{7}\sum_k E_k^{\mathrm{LOO}}$} & \textbf{10.17} & \\
\multicolumn{3}{@{}l}{LOO average without LRG2 (six folds)} & 2.54 & \\
\multicolumn{4}{@{}l}{\textbf{Per-bin-independent product} $\prod_k M_k^{\mathrm{REGROW},\delta=2}$} & \textbf{4.46} \\
\multicolumn{4}{@{}l}{Look-elsewhere-corrected mean $\bar M = \tfrac{1}{7}\sum_k M_k^{\mathrm{REGROW},\delta=2}$} & 1.52 \\
\bottomrule
\end{tabular}
\caption{Leave-one-out and per-bin-independent e-values by redshift bin, at effective redshift $z_{\mathrm{eff}}$. \textit{Columns 3--4 (LOO)}: in each fold, $(w_0, w_a)$ is fitted to the six remaining bins and the resulting cosmology is used to compute $E_k^{\mathrm{LOO}}$ on the held-out bin (Proposition~\ref{prop:loo}). \textit{Column 5 (per-bin mixture)}: mixture e-value of bin $k$ alone under the REGROW $\delta = 2$ test specification; the aggregate rows are defined in the text. The two columns score different alternatives and are not on a common scale. LRG2, the LOO average and the per-bin-independent product are set in bold.}
\label{tab:loo}
\end{table}

Table~\ref{tab:loo} decomposes the DR2 evidence bin by bin in two ways, and the two columns score different alternatives, so they are not on a common scale. The leave-one-out (LOO) column ($E_k^{\mathrm{LOO}}$) holds out each bin in turn and tests it against the cosmology fitted to the remaining six. Each entry $M_k$ in the per-bin mixture column scores a single bin on its own under the REGROW $\delta = 2$ prior. The per-bin product and mean under the Default prior, which Bartlett dilution dominates (Section~\ref{sec:prior-sens}), are in Appendix~\ref{app:loo-decomp}. The column's arithmetic mean $\bar M$ corrects, by linearity of expectation, for selecting the most favourable of the seven bins, and its product $\prod_k M_k$ is a valid e-value for the more flexible per-bin-independent alternative, in which each bin may favour its own $(w_0, w_a)$ (Appendix~\ref{app:loo-decomp}). LRG2 reads $55.98$ in one column and $4.25$ in the other. The gap opens because REGROW $\delta = 2$ dilutes single-bin evidence through its Occam penalty (Section~\ref{sec:prior-sens}), whereas LOO scores the held-out bin against the sharp six-bin best fit. We therefore read each column against its own aggregate row.

Both decompositions localise the DR2 rejection to a single bin, LRG2 (Figure~\ref{fig:loo_bins}). LRG2 carries $78.6\%$ of $\sum_k E_k^{\mathrm{LOO}}$ in the LOO column ($90.9\%$ together with LRG3$+$ELG1). In the per-bin mixture column $M_{\mathrm{LRG2}} = 4.25$ is the largest value, and the next, LRG3$+$ELG1, is $1.84$. The six-bin MLEs also stay stable from fold to fold ($w_0 \in [-0.88, -0.82]$, $w_a \in [-0.60, -0.33]$), so the folds agree on the direction of the fit and differ in how much evidence each carries.

The rejection does not survive removing LRG2 or correcting for having inspected all seven bins. The LOO average over the six folds other than LRG2 is $2.54$ against $\bar E^{\mathrm{LOO}} = 10.17$ over all seven, and the DR2 mixture recomputed on the remaining six bins under the Default prior drops from $M_{\mathrm{DR2}} = 33.97$ to $0.49$, so the six residual bins marginally \emph{favour} \LCDM. The look-elsewhere-corrected mean $\bar M = 1.52$ does not reach $20$, so a per-bin alternative corrected for having inspected all seven bins does not reject, and the per-bin-independent product $\prod_k M_k = 4.46$ does not reach $20$ either. These results reproduce the bin-removal findings of \cite{Liu2024LRG, Wang2024DR1, ColgainEvolved2025, Chaudhary2025} under a formally calibrated statistic, and they agree with an out-of-sample predictive analysis in which LRG2 alone supplies the BAO preference for $w_0w_a$CDM~\cite{KarimSpeagleHlozek2026}.

\LCDM\ does not produce concentration this strong, but concentration alone does not show that the signal is localised rather than smooth. We separate the two readings with a parametric bootstrap (Appendix~\ref{app:loo-decomp}). No simulated realisation under \LCDM\ truth placed $75\%$ of the LOO evidence in a single bin, so the observed concentration is a real departure from the null. A genuinely smooth $w_0w_a$CDM signal, by contrast, produces such concentration often, because the bins differ widely in redshift lever arm and precision. The two readings differ instead in how much of the evidence one particular bin, LRG2, carries. A smooth signal concentrates most often in LRG3$+$ELG1, and it reproduces the observed LRG2 share in only $3.2\%$ of realisations. The evidence therefore sits in a bin where a smooth equation of state rarely puts this much of it, and that placement is consistent with an outlier or a bin-specific systematic.

LRG2's values also changed most between releases. Its transverse distance has the largest of the eleven pulls in the DR2 minus DR1 null test of Appendix~\ref{app:hierarchical}, $+1.96\sigma$ at $\alpha_1 = 1/3$ (for the bin, $\chi^2 = 4.1$ for 2 degrees of freedom, $p = 0.13$). DESI's own comparison of the two releases reports its largest differences in LRG2 ($1.93\sigma$ in the Alcock--Paczynski parameter $\alpha_{\mathrm{AP}}$) and LRG3$+$ELG1 ($1.56\sigma$ in the isotropic scale $\alpha_{\mathrm{iso}}$), and it finds the releases consistent (Kolmogorov--Smirnov $p = 0.40$)~\citep{DESI2025}. {\'O}~Colg{\'a}in et al.~\citep{ColgainEvolved2025} noted the change in LRG2's preferred $\Omega_m$ between releases. One expects LRG2 to move most if the DR2 excess comes from data or processing new in DR2. A $2\sigma$ pull among eleven is unremarkable, however, so the shift is consistent with the fragility of the signal and is not evidence of a systematic.

\begin{figure}[t]
\centering
\includegraphics[width=\textwidth]{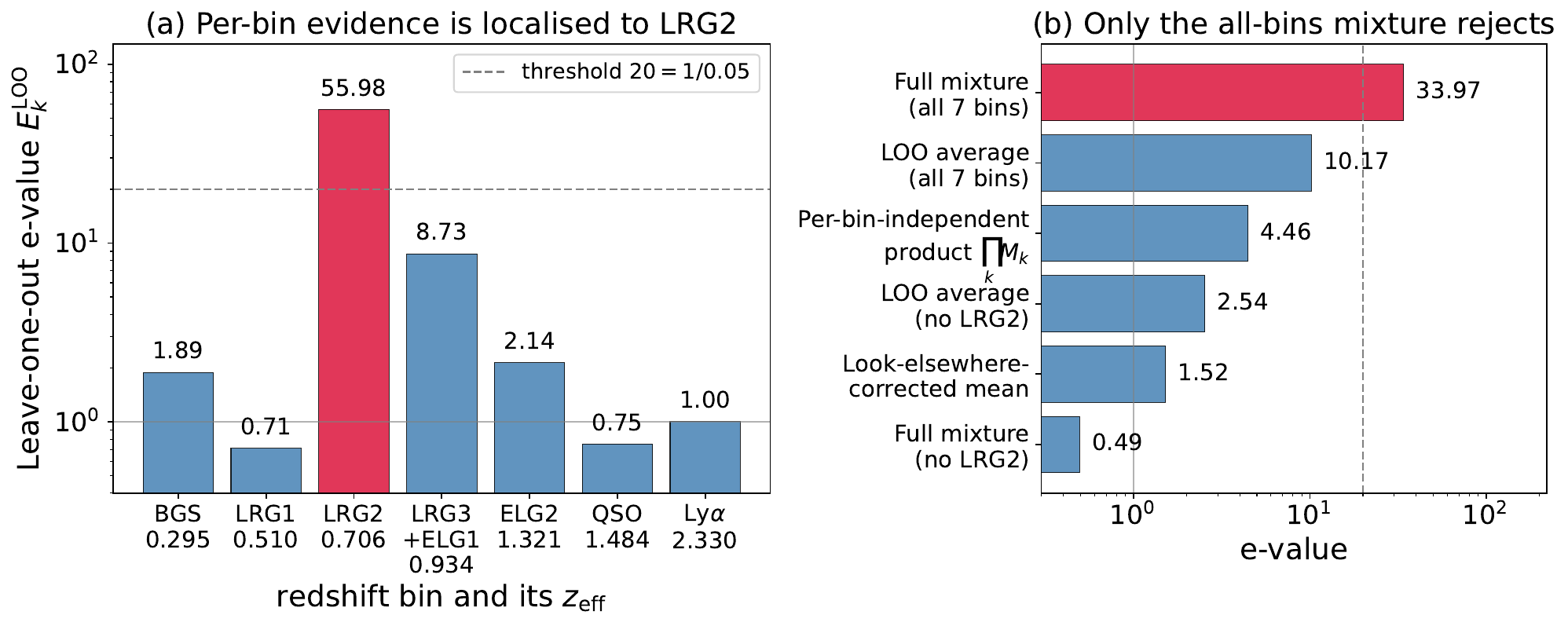}
\caption{Per-bin leave-one-out (LOO) e-values for DESI DR2 BAO. \emph{(a)} LOO e-value $E_k^{\mathrm{LOO}}$ of each redshift bin, tested against $w_0w_a$CDM fitted to the other six (Table~\ref{tab:loo}). \emph{(b)} Aggregate e-values against the same threshold: only the all-bins DR2 mixture rejects (red); the other aggregates, including the values without LRG2 and the look-elsewhere-corrected mean, all fall below.}
\label{fig:loo_bins}
\end{figure}

\subsection{Sensitivity to the test specification and prior}
\label{sec:prior-sens}

Each $\pi$ defines its own e-value, so we must fix $\pi$ in advance of any decision. Table~\ref{tab:summary} shows how much that choice matters. The rejection survives the Narrow and Default boxes and REGROW $\delta \geq 2$, but it fails under the Wide and Ong et al.\ boxes and REGROW $\delta = 1$. In other words, an analyst who targets deviations of at least $2$ Fisher-$\sigma$ can reject \LCDM\ from DR2, and one who targets $1$ Fisher-$\sigma$ cannot. The Bayes-factor analyses report the same prior sensitivity \citep{OngBayesian2025, WangMota2025}, but the e-value attaches a Type-I guarantee to whichever specification the analyst chooses.

The geometry of the DR2 likelihood explains the pattern. The MLE sits $\delta_{\mathrm{MLE}} = 4.06$ Fisher-$\sigma$ from \LCDM, so a specification gains evidence as it moves mass towards that region. The closed-form Gaussian-shell proxy $\pi_\delta = \mathcal N(\theta_0, \delta^2 F^{-1})$ has a data-supported window of $\delta \in [0.87, 4]$, but that window belongs to the Gaussian-shell family alone and does not transfer to the REGROW family of Table~\ref{tab:summary} (Appendix~\ref{app:delta-star}). The wide boxes fail for the complementary reason, Bartlett--Jeffreys--Lindley dilution. A diffuse prior spreads its mass over regions the data cannot distinguish, and it pays for that spread through the Occam factor $(1+\delta^2)^{-k/2}$ of equation~\eqref{eq:M-delta} in Appendix~\ref{app:delta-star}. Any wide-prior method, frequentist or Bayesian, pays that penalty, so it is not peculiar to e-values.

The thawing and freezing priors replace the reference boxes with distributions that dark-energy models predict, and the data separate the two sharply. Late-time quintessence with a standard kinetic term divides into freezing potentials (tracker, SUGRA; broadly $w_0 \in [-0.95, -0.75]$, $w_a \in [0, 0.3]$) and thawing potentials (PNGB, linear). The thawing potentials in slow-roll occupy the Caldwell--Linder band $w_a \in [{-3(1+w_0)},\,{-(1+w_0)}]$ for $w_0 \in [-1, -0.85]$ \citep{CaldwellLinder2005, ScherrerSen2008}, and simple thawing models that satisfy the null energy condition predict a preference for exactly this sector \citep{ShlivkoSteinhardt2024}. We integrate \eqref{eq:Mt} against each region in turn, and the freezing region gives $M_{\mathrm{DR2}}^{\mathrm{freezing}} = 14.8$ (on a $320 \times 320$ grid), which does not reject. The thawing band, by contrast, reaches $M_{\mathrm{DR2}}^{\mathrm{thawing}} \approx 1.1 \times 10^3$ ($0.98$ to $1.13 \times 10^3$ as the band's upper $w_0$ edge moves between $-0.80$ and $-0.85$; $M_{\mathrm{DR1}}^{\mathrm{thawing}} \approx 15$). It scores strongly because it runs nearly along the major, least-constrained axis of the anisotropic DR2 Fisher information. The band's steep-thawing edge passes through the DR2-favoured region and contains the BAO MLE $(w_0, w_a) = (-0.856, -0.430)$, so the prior stays where the likelihood is high and pays almost no Occam penalty.

\subsection{Forecasting DR3 and beyond}
\label{sec:forecast}

We call DR3 the next public DESI BAO release, expected in 2027,\footnote{DESI completed its planned five-year survey on 16 April 2026 and expects the first dark energy results from the full data set in 2027 (Fermilab news release, 17 April 2026).} and we now fix how it enters. DR3 joins the weighted average of Proposition~\ref{prop:weighted}, and we weight each planned release of the five-year programme in proportion to its observing time, because observing time measures the information a release carries. DR1, DR2 and DR3 (about one, three and five years) therefore receive $1/9$, $3/9$ and $5/9$. Observing time only sets the proportions, because any weights fixed in advance are valid (Proposition~\ref{prop:weighted}). Group-sequential designs set theirs in much the same way, when they spend their error budget across planned looks by the information fraction reached at each~\citep{LanDeMets1983}. The published DR2 variances, about a third of DR1's (ratios $0.22$ to $0.49$, median $0.32$), bear out observing time as the measure. The weights sum to one, so releases after the five-year programme lie outside this rule and would need one of their own, stated before their data.

The running sum at DR2 is $S_{\mathrm{DR2}} = \sum_{t \le \mathrm{DR2}} w_t M_t = 11.44$ ($p \le 0.0874$), so DR3 needs $M_{\mathrm{DR3}}$ of about $15$ for the sum to reach $20$. Like the equal weights, this rule comes after DR2 was public, but it fixes DR3's share in advance. The weights follow the survey design rather than the observed values, and the guarantee holds for the sequence as a whole, not conditional on the DR2 value already seen. DR3 will also reprocess the DR2 data, so a DR3 crossing driven by the same LRG2 measurements would not be independent confirmation.

What DR3 can settle depends on which measurements improve, and the decisive ones are more precise measurements in the existing bins, above all LRG2. We simulated plausible fourth- and fifth-year data to see which measurements matter. Two new high-redshift bins at $z = 1.7$ and $2.5$ barely move the e-value under either hypothesis, because \LCDM\ and the DESI-preferred $w_0w_a$CDM predict distances there that differ by less than $0.5\sigma$ at the projected precision. Halving the LRG2 error in the same simulations, by contrast, sends the median e-value to about $1.7 \times 10^4$ if $w_0w_a$CDM is true and leaves it near zero if \LCDM\ is true. These forecasts assume that the published DR2 covariance scales with exposure at the fixed background of Section~\ref{sec:methods}, and they indicate which measurements would be decisive rather than a value of $M_{\mathrm{DR3}}$.

\section{Discussion and conclusions}
\label{sec:discussion}

An analyst can combine evidence across superseding releases with a guarantee, using public products alone. The analyst averages the per-release e-values with weights fixed in advance, and that average controls false detections at every release, whatever the relation between the releases, provided each release's own likelihood is correct. We did not fix the equal weights of our average before DR2 was published. The weights we fix now, ahead of DR3 and in proportion to observing time, give $11.44$ at DR2 ($p \le 0.0874$).

Within DR2 the evidence clears the threshold but is fragile, because it rests on one bin. That bin is LRG2, at $z_{\text{eff}} = 0.706$, and without it the DR2 e-value of $33.97$ falls to $0.49$. The equal-weight average of DR1 and DR2 counts the DR1 look as well and stands at $17.51$, so it does not reach $20$ ($p \le 0.0571$). A model of how the two releases are related does not pin the value down either. It gives anything from no evidence to nearly double the DR2 value, depending on the structure assumed (Section~\ref{sec:eprocess}, Appendix~\ref{app:hierarchical}). The one informed reference case uses the correlations DESI adopts for its own DR1--DR2 consistency check, and it gives about the DR2 value ($28.9$ to $33.7$). Those correlations, though, are a heuristic for a consistency check rather than a published covariance. The evidence therefore amounts to one strong look at DR2, driven by one bin, and a model-based combination cannot sharpen it from public products.

The rejection rests on two conditions at once. The test specification must concentrate near the data's preferred direction, and the redshift bins must share the signal coherently. A wider prior relaxes the first condition, and per-bin statistics, which let each bin favour its own $(w_0, w_a)$ or correct for having inspected all seven bins, relax the second. Either step makes the rejection disappear (Sections~\ref{sec:loo} and~\ref{sec:prior-sens}). Simulations also show that a smooth equation of state rarely concentrates this much of the evidence in LRG2 specifically (Appendix~\ref{app:loo-decomp}). Only additional data will show whether LRG2 reflects new physics or a localised systematic.

Four caveats qualify these conclusions, and Appendix~\ref{app:limitations} sets out each one with the follow-up it points to. We hold the background cosmology fixed at Planck 2018 values, and Table~\ref{tab:background-sensitivity} shows the headline e-values moving by two orders of magnitude across self-consistent Planck columns. Every number also conditions on the CPL parametrisation, so the most direct check still open is whether the LRG2 signal survives outside CPL. A single data-split at $z = 1$, which we report only in Appendix~\ref{app:limitations}, is underpowered for $w_a$, so we do not rely on it. Finally, the cross-dataset comparison scores pre-specified SN$+$CMB alternatives after the fact rather than tracking the surveys jointly, and joint tracking would need the full constituent likelihoods rather than the published posteriors used here. Beyond these four caveats, the choice between two published forms of the DR2 covariance also moves $M_{\mathrm{DR2}}$ by a factor of two (Section~\ref{sec:methods}).

We recommend that each release of a multi-release survey report $M_t$ with a stated test specification $\pi$, alongside the standard Wilks $\sigma$ and a prior-sensitivity table like Table~\ref{tab:summary}. We also recommend that the survey fix, in advance, the rule for combining it with earlier and later releases.

An exact combination across releases needs the cross-release covariance (Section~\ref{sec:eprocess}, Appendix~\ref{app:hierarchical}). The cheapest product that would supply it is the BAO fits of the DR1 and DR2 pipelines on the same phase-matched mocks, or the resulting cross-release covariance, released with each BAO release. Anyone could then test release-to-release shifts such as LRG2's $1.93\sigma$ in $\alpha_{\mathrm{AP}}$. Outside analysts cannot yet build it themselves, because as of October 2026 the DR1 large-scale-structure (LSS) catalogues and mocks and both pipelines are public but the DR2 LSS catalogues and mocks are not.\footnote{\url{https://data.desi.lbl.gov/doc/releases/}, read 5 October 2026, lists the Early Data Release and DR1 as data releases; DR2 appears there only through its cosmology chains and best-fit products.} The need for evidence that stays calibrated when releases supersede one another grows with data volume and reanalysis frequency along the DESI$\to$Euclid$\to$Roman$\to$LSST chain.

\appendix
\section{E-value theory and proofs}
\label{app:proofs}

\begin{definition}[E-value]
A random variable $E \geq 0$ with ${\E[E \mid H_0] \leq 1}$ is an e-value for testing $H_0$.
\end{definition}

\begin{theorem}[Ville's inequality]
\label{thm:ville}
Let $\{Z_t\}_{t \geq 0}$ be a non-negative supermartingale under $H_0$ with $\E[Z_0 \mid H_0] \leq 1$. Then for any $\alpha \in (0,1)$,
\begin{equation}
\Prob_{H_0}\!\left(\sup_{t \geq 0} Z_t \geq 1/\alpha\right) \leq \alpha.
\end{equation}
\end{theorem}

\begin{proof}
This is the maximal inequality for non-negative supermartingales~\cite{Ville1939}: with $\tau$ the first time $Z_t \geq 1/\alpha$, optional stopping gives $\E[Z_{t \wedge \tau}] \leq 1$, Markov's inequality on the event $\{\tau \leq t\}$ gives $\Prob(\tau \leq t) \leq \alpha$, and $t \to \infty$ completes the argument.
\end{proof}

\noindent Ville's inequality does not apply to the per-release statistic for DESI's superseding releases. The mixture of \eqref{eq:Mt} is a valid e-value at each fixed $t$, but it is not a supermartingale across releases, so the theorem's hypothesis fails. The theorem does hold for the model-based construction of Appendix~\ref{app:hierarchical}, though only under that appendix's cross-release model. Proposition~\ref{prop:weighted} covers the combination that needs no cross-release model.

\begin{corollary}[Single-look Markov bound]
\label{cor:markov}
If $E$ is an e-value, then ${\Prob_{H_0}(E \geq 1/\alpha) \leq \alpha}$.
\end{corollary}

\begin{proof}
By Markov's inequality, $\Prob(E \geq 1/\alpha) \leq \alpha\,\E[E] \leq \alpha$. This is the $t = 0$ special case of Theorem~\ref{thm:ville}.
\end{proof}

\begin{proposition}[Likelihood-ratio e-value]
\label{prop:lr}
For simple hypotheses $H_0: P = P_0$ versus $H_1: P = P_1$, the likelihood ratio $E = P_1(X)/P_0(X)$ satisfies $\E_{P_0}[E] = 1$.
\end{proposition}

\begin{proof}
$\E_{P_0}[E] = \int \frac{P_1(x)}{P_0(x)} P_0(x) \, dx = \int P_1(x) \, dx = 1$.
\end{proof}

\begin{proposition}[Data-split e-value validity]
\label{prop:split}
Let $D = (D_{\mathrm{train}}, D_{\mathrm{test}})$ be independent data splits. Let $\hat{\theta} = \hat{\theta}(D_{\mathrm{train}})$ be parameters fitted to the training data. Then
\begin{equation}
E = \frac{L(D_{\mathrm{test}} \mid \hat{\theta})}{L(D_{\mathrm{test}} \mid H_0)}
\end{equation}
is a valid e-value for testing $H_0$.
\end{proposition}

\begin{proof}
Conditional on $D_{\mathrm{train}}$, $\hat{\theta}$ is fixed. Under $H_0$, $D_{\mathrm{test}} \sim P_0$ independent of $D_{\mathrm{train}}$. Thus
\[
\E[E \mid D_{\mathrm{train}}] = \int \frac{L(x \mid \hat{\theta})}{L(x \mid H_0)} P_0(x) \, dx = 1
\]
by the same argument as Proposition~\ref{prop:lr}. Taking expectations, ${\E[E] = \E[\E[E \mid D_{\mathrm{train}}]] = 1}$.
\end{proof}

\begin{proposition}[Uniform-mixture e-value validity]
\label{prop:mixture}
Let $\mathcal{G} = \{\theta_1, \ldots, \theta_J\}$ be a finite set of parameter values that does not depend on the data. The uniform mixture
\begin{equation}
E_{\mathrm{mix}} = \frac{1}{J} \sum_{j=1}^J \frac{L(\mathrm{data} \mid \theta_j)}{L(\mathrm{data} \mid H_0)}
\end{equation}
is a valid e-value.
\end{proposition}

\begin{proof}
Each term $E_j = L(\mathrm{data} \mid \theta_j) / L(\mathrm{data} \mid H_0)$ is a likelihood ratio for a simple alternative, so $\E_{H_0}[E_j] = 1$ by Proposition~\ref{prop:lr}. By linearity of expectation,
\[
\E_{H_0}[E_{\mathrm{mix}}] = \frac{1}{J} \sum_{j=1}^J \E_{H_0}[E_j] = \frac{1}{J} \cdot J = 1. \qedhere
\]
\end{proof}

\begin{proposition}[LOO-average e-value validity]
\label{prop:loo}
Let $E_1, \ldots, E_K$ be leave-one-out e-values, where $E_k = L(D_k \mid \hat{\theta}_{-k}) / L(D_k \mid H_0)$ and $\hat{\theta}_{-k}$ is fitted on all data except fold $k$. The average $\bar{E} = \frac{1}{K}\sum_{k=1}^K E_k$ is a valid e-value.
\end{proposition}

\begin{proof}
Each $E_k$ is a valid e-value by Proposition~\ref{prop:split}: conditional on the training data $D_{-k}$, $\hat{\theta}_{-k}$ is fixed, and the held-out bin $D_k$ is independent of the training bins because the published covariance is block-diagonal across bins and the likelihood is Gaussian. Therefore ${\E[E_k \mid H_0] \leq 1}$ for each $k$. By linearity,
\[
\E[\bar{E} \mid H_0] = \frac{1}{K}\sum_{k=1}^K \E[E_k \mid H_0] \leq \frac{1}{K} \cdot K = 1. \qedhere
\]
\end{proof}

\noindent The \textit{product} $\prod_k E_k$ does \textit{not} share this guarantee for LOO folds, because the training sets overlap and the fitted $\hat\theta_{-k}$ are therefore correlated across folds. Section~\ref{sec:loo} instead multiplies per-bin mixture e-values, $\prod_k M_k^{\mathrm{REGROW}}$. That product is valid, because each factor uses a fixed prior rather than LOO-fitted parameters and because the block-diagonal DR2 covariance makes the bins independent.

\begin{proposition}[Weighted average across releases]
\label{prop:weighted}
Let $M_1, M_2, \ldots$ be non-negative statistics, each computed from a single release with that release's published likelihood, such that ${\E[M_t \mid H_0] \leq 1}$ for every $t$. Let $w_1, w_2, \ldots \geq 0$ be deterministic weights, fixed independently of all release data, with $\sum_t w_t \leq 1$. Then $S_T = \sum_{t \leq T} w_t M_t$ satisfies $\E[S_\tau \mid H_0] \leq 1$ for every stopping time $\tau$, and
\begin{equation}
\Prob_{H_0}\!\left(\sup_{T \geq 1} S_T \geq 1/\alpha\right) \leq \alpha,
\end{equation}
whatever the joint distribution of the releases.
\end{proposition}

\begin{proof}
Because $M_t \geq 0$ and $w_t \geq 0$, $S_T$ is non-decreasing in $T$, so $S_\tau \leq \sup_T S_T = S_\infty = \sum_t w_t M_t$. By linearity of expectation, which needs no assumption about the dependence between the $M_t$, and monotone convergence, $\E[S_\infty \mid H_0] = \sum_t w_t\, \E[M_t \mid H_0] \leq \sum_t w_t \leq 1$. Markov's inequality applied to $S_\infty$ gives the bound.
\end{proof}

\noindent The weights may not depend on earlier releases. Under arbitrary dependence, a weight chosen after seeing release $t-1$ can favour a later release that is predicted to be large, and the bound then fails.

\noindent The guarantee is also unconditional. Unlike a supermartingale, this proposition says nothing conditional on earlier releases. The distinction matters because later releases reuse earlier exposures.

\begin{remark}[Optimality of the average]
\label{rem:vovk-wang}
A function that is non-decreasing in each argument and maps $K$ e-values to a number that is again an e-value under every dependence structure is an e-merging function. Vovk and Wang~\citep[Proposition~3.1]{Vovk2021} show that the arithmetic mean essentially dominates every symmetric e-merging function: any other symmetric rule valid under arbitrary dependence is no larger than the mean wherever it exceeds one, the region that matters for rejection (see also~\citep{RamdasWang2025}). By Theorem~3.2 of the same paper, every symmetric e-merging function is dominated by $\lambda + (1-\lambda)$ times the mean for some $\lambda \in [0, 1]$. More generally, every admissible e-merging function is a weighted arithmetic average~\citep{Wang2024}, and weighted averages with unequal weights are valid (Proposition~\ref{prop:weighted}) provided the weights are fixed in advance. Symmetry is therefore a convention rather than a consequence of supersession: weights that favour the later release give a stronger combined statement for DR1 and DR2, and weights that keep some for future releases a weaker one (Section~\ref{sec:eprocess}).
\end{remark}

\begin{proposition}[Worst case over a family of cross-release models]
\label{prop:family}
Let $\{P_a : a \in A\}$ be a family of joint laws for the releases under $H_0$, which must not be narrowed after the data are seen, and for each $a \in A$ let $M(a) \geq 0$ satisfy $\E_{P_a}[M(a)] \leq 1$. If the true joint law is $P_{a^*}$ for some $a^* \in A$, then $M_{\inf} = \inf_{a \in A} M(a)$ satisfies $\E_{P_{a^*}}[M_{\inf}] \leq 1$; it is an e-value for the composite null $\{P_a : a \in A\}$.
\end{proposition}

\begin{proof}
$0 \leq M_{\inf} \leq M(a^*)$ pointwise, so $\E_{P_{a^*}}[M_{\inf}] \leq \E_{P_{a^*}}[M(a^*)] \leq 1$.
\end{proof}

\noindent The guarantee covers only laws inside the family, and it is informative only if the worst case over the family is finite and positive. A family that reaches or approaches a degenerate joint law can lose that: over the class of all joint Gaussian laws consistent with the published marginals, for example, the combined statistic can be made arbitrarily small or arbitrarily large (Appendix~\ref{app:hierarchical}). On the DESI data the one-fraction family of that appendix keeps it, because there its statistic diverges upward at the degenerate limit, so the minimum is attained inside the range. We apply the proposition to the noise-fraction family of Appendix~\ref{app:hierarchical}, so the result inherits that family's structural assumptions and says nothing about relations outside it. Taking the worst case over a composite null is nonetheless the standard construction~\citep{Ramdas2023, RamdasRufLarssonKoolen2022}. The family may not be narrowed after the data are seen, and the answer depends on the family's dimension as well as its form (Appendix~\ref{app:hierarchical}). Appendix~\ref{app:fixed-planck} applies the same construction to the choice of fixed background, where the family is a set of single-release laws, and the argument is unchanged.

\begin{remark}[Gaussian likelihoods]
For a multivariate Gaussian likelihood of the data vector $x$ with known covariance $C$, prediction $\mu(\theta)$ and a fixed alternative $\theta_1$,
\begin{equation}
E = \exp\left(\frac{\chi^2(H_0) - \chi^2(\theta_1)}{2}\right)
\end{equation}
where $\chi^2(\theta) = (x - \mu(\theta))^T C^{-1} (x - \mu(\theta))$.
\end{remark}

\begin{remark}[Maximised $\chi^2$ is not an e-value]
Replacing the fixed $\theta_1$ above by the maximiser over a $k$-dimensional alternative, $\Delta\chi^2(k) = \sup_\theta[\chi^2(H_0) - \chi^2(\theta)]$, breaks the e-value property. For $k \geq 1$, $\Delta\chi^2(k) \sim \chi^2_k$ under $H_0$ (Wilks), so
\begin{equation}
\E_{H_0}\!\left[\exp\!\left(\tfrac{1}{2}\Delta\chi^2(k)\right)\right]
= \int_0^\infty e^{x/2}\,\frac{x^{k/2-1}e^{-x/2}}{2^{k/2}\Gamma(k/2)}\,dx
= \infty,
\label{eq:wilks-divergence}
\end{equation}
the integrand reducing to $x^{k/2-1}/(2^{k/2}\Gamma(k/2))$, which is not integrable. The null expectation diverges, so $e^{\Delta\chi^2/2}$ is not an e-value. This is why the mixture of \eqref{eq:Emix} averages over $\theta$ rather than maximising.
\end{remark}

\section{Methodological notes}
\label{app:methods-notes}

\subsection{Why we do not report a \texorpdfstring{$\sigma$}{sigma}-conversion of the e-value}
\label{app:sigma}
We do not convert the e-value into an equivalent Gaussian $\sigma$, because an e-value has no single null distribution. Its defining constraint, $\E[E \mid H_0] \leq 1$, yields only Markov's inequality, $\Prob(E \geq k \mid H_0) \leq 1/k$, so the definition does not pin down a conversion. Any conversion returns a range instead, and the width of that range depends on how much structure one assumes beyond the constraint. The strict Markov bound $\sigma_{\mathrm{Markov}} = \Phi^{-1}(1 - 1/(2E))$, with $\Phi^{-1}$ the standard normal quantile function, assumes nothing more. It gives $2.18$ at $E = 33.97$, but it understates the frequentist precision of the parametric model. A Monte Carlo tail computed under the Gaussian likelihood DESI assumes ($N = 2 \times 10^5$ draws, code release) marks the other end of the range. That tail gives $\sigma_{\mathrm{emp}} = 3.69$ (tail probability $2.25 \times 10^{-4}$) and reproduces Wilks in this aligned-prior, local-Gaussian regime. That number, however, gives up the Type-I control that rests on $\E[E \mid H_0] \leq 1$ alone and motivated the e-value, and it describes a single look. We therefore report $E$ itself, with its Markov inversion $1/E$, and introduce no single ``e-value $\sigma$''.

\subsection[Cross-release combination: noise-fraction family and calibration budget]{Cross-release combination: noise-fraction family, null test and calibration budget}
\label{app:hierarchical}

The published products of DESI's two releases do not determine their joint distribution, so any combination that scores each observation once rests on a model of how the releases are related. This appendix defines families of such models and prices each assumption.

DR2 reprocesses the DR1 observations rather than adding two more years of exposure to a fixed DR1. DESI reduced the whole DR2 data set, DR1 epochs included, together with the new observations in one homogeneous processing run (``Loa'')~\citep{DESI2025}. It also redefined the BGS sample ($M_r < -21.35$ in place of $-21.5$), raised the minimum fitting scale from $50$ to $60\,h^{-1}\,\mathrm{Mpc}$, moved ELG1 and the quasars from isotropic to anisotropic fits, and reassessed parts of the systematic-error budget~\citep{Andrade2025}. The collaboration describes the pipeline as largely unchanged~\citep{Andrade2025}, but none of this information fixes the cross-release covariance.

The one-fraction family assumes that a fixed fraction $\alpha_1$ of the DR1 noise persists into DR2, the same in every bin. We write $X_1$ and $X_2$ for the DR1 and DR2 measurements of the 11 quantities that both releases report, namely $D_V$ for BGS, and $D_M$ and $D_H$ for LRG1, LRG2, LRG3$+$ELG1, ELG2 and Ly$\alpha$. The QSO bin is unmatched, because DR1 reports only $D_V$ there. The family is
\begin{equation}
X_1 = \mu_1 + \varepsilon_{1}, \qquad X_2 = \mu_2 + \alpha_1\,\varepsilon_{1} + (1-\alpha_1)\,\varepsilon_{23}, \qquad \varepsilon_{23} \perp \varepsilon_{1},
\label{eq:family}
\end{equation}
where $\mu_1$ and $\mu_2$ are the BAO observables predicted at each release's own effective redshifts, the same function of the cosmology for both, $\varepsilon_1$ has covariance $C_{\mathrm{DR1}}$, and $\alpha_1 \geq 0$. The model then gives $\mathrm{Cov}(X_1, X_2) = \alpha_1 C_{\mathrm{DR1}}$, and matching the published DR2 covariance requires
\begin{equation}
\mathrm{Cov}(\varepsilon_{23}) = \frac{C_{\mathrm{DR2}} - \alpha_1^2\, C_{\mathrm{DR1}}}{(1-\alpha_1)^2},
\end{equation}
which is positive semi-definite only for $\alpha_1 \leq 0.4701$, the square root of the smallest generalised eigenvalue of $C_{\mathrm{DR2}}$ relative to $C_{\mathrm{DR1}}$. At that limit the matrix is singular, so a single common fraction may take any value in $0 \leq \alpha_1 < 0.4701$. The survey durations suggest $\alpha_1 = 1/3$, but the published marginals do not identify it.

The family also sets the cross-release covariance of the QSO measurements to zero (DR1 $D_V$ against DR2 $D_M$ and $D_H$), although DR2 contains the DR1 quasars. This assumption does not drive the result, since dropping the DR2 QSO term raises the $\alpha_1 = 1/3$ statistic from $66.7$ to $90.2$.

The model-based combined statistic scores DR1 once and then scores only the information DR2 adds. For each $\alpha_1$ the part of DR2 that is independent of DR1 is $y = (X_2 - \alpha_1 X_1)/(1-\alpha_1)$, with mean $(\mu_2 - \alpha_1\mu_1)/(1-\alpha_1)$ and covariance $\mathrm{Cov}(\varepsilon_{23})$. We take the unmatched DR2 quasar measurements $X_2^{\mathrm{U}}$ to be independent of DR1, and the statistic is then
\begin{equation}
M_{\mathrm{DR2}}^{\mathrm{joint}}(\alpha_1) = \int \frac{L(\mathcal D_{\mathrm{DR1}} \mid \theta)\, L(y \mid \theta)\, L(X_2^{\mathrm{U}} \mid \theta)}{L(\mathcal D_{\mathrm{DR1}} \mid H_0)\, L(y \mid H_0)\, L(X_2^{\mathrm{U}} \mid H_0)}\, \pi(\theta)\, d\theta,
\label{eq:Mjoint}
\end{equation}
where the DR1 factor uses the full 12-quantity DR1 vector $\mathcal D_{\mathrm{DR1}}$ of \eqref{eq:Mt}, the DR1 quasar $D_V$ included, while the cross-covariance acts on the 11 matched quantities. Under $H_0$ and the corresponding member of the family, the statistic satisfies $\E[M_{\mathrm{DR2}}^{\mathrm{joint}} \mid \mathcal F_1] = M_{\mathrm{DR1}}$, where $\mathcal F_1$ is the information in DR1. The identity holds because the conditional factor integrates to one for each $\theta$ by the Gaussian moment identity, while the DR1 factor is $\mathcal F_1$-measurable. At $\alpha_1 = 1/3$ the data give $M_{\mathrm{DR2}}^{\mathrm{joint}} = 66.7$.

The $\alpha_1 = 1/3$ construction and the scan over $\alpha_1 \in [0.25, 0.40]$ date from the first two arXiv versions of this paper. We added the other analyses of this appendix afterwards, and they leave the per-release results unchanged.

The statistic never falls below $64.7$ across the one-fraction family, so it rejects for every common fraction (Figure~\ref{fig:checks}a). It starts at $1025$ at $\alpha_1 = 0$, falls to a broad minimum of $64.7$ at $\alpha_1 = 0.31$, and diverges as $\alpha_1 \to 0.4701$, where the covariance of the new information becomes singular. Proposition~\ref{prop:family} shows that the minimum is an e-value whenever the true law is a member of the family, which gives $p \leq 0.0155$. The minimum matters because no single member is protected against the others. If the true fraction were $0.25$, for example, the $\alpha_1 = 1/3$ statistic would have null expectation $162$, and that expectation would exceed $10^{13}$ if the releases were independent.

\begin{figure}[t]
\centering
\includegraphics[width=\textwidth]{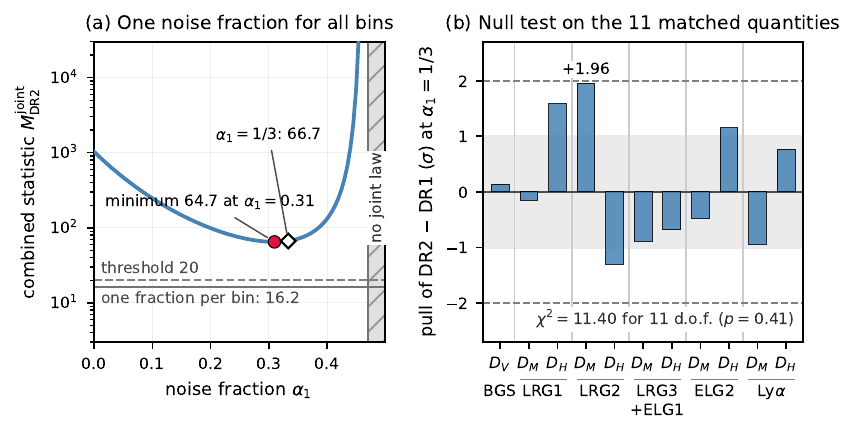}
\caption{Checks on the cross-release families. \emph{(a)} The model-based combined statistic for the one-fraction family $\mathrm{Cov}(X_1, X_2) = \alpha_1 C_{\mathrm{DR1}}$. The red point is the family's minimum, $64.7$ at $\alpha_1 = 0.31$, and the diamond is $\alpha_1 = 1/3$ ($66.7$); beyond $0.4701$ (hatched) no joint law exists. The dashed line is the threshold $20$ and the solid grey line the worst case with one fraction per bin ($16.2$). \emph{(b)} Pulls of the DR2$-$DR1 differences of the 11 matched quantities at $\alpha_1 = 1/3$, corrected for the effective-redshift shifts of two bins. The band marks $\pm 1$ and the dashed lines $\pm 2$; $\chi^2 = 11.40$ for 11 degrees of freedom ($p = 0.41$).}
\label{fig:checks}
\end{figure}

The one-fraction minimum stops being conservative as soon as the fraction may differ between bins. DESI's own DR1--DR2 consistency check supplies a first bin-dependent variant. That check takes the correlation between a bin's DR1 and DR2 measurements to be $r = (N_{\mathrm{DR1}}/N_{\mathrm{DR2}})^{1/2}$ for galaxies and quasars, with $N$ the number of tracers, and $r = 0.61$ for the Ly$\alpha$ forest (written $C$ there). It quotes the joint uncertainty $(\sigma_{\mathrm{DR2}}^2 + \sigma_{\mathrm{DR1}}^2 - 2r\,\sigma_{\mathrm{DR1}}\,\sigma_{\mathrm{DR2}})^{1/2}$~\citep[Sec.~III.3.1 and footnote~12]{DESI2025}. The printed formula omits a factor $\sigma_{\mathrm{DR2}}$ in the cross term. The rule serves as a heuristic for a consistency check rather than as a cross-release covariance, but it still implies a combined statistic. For each matched quantity $i$ it corresponds to $\alpha_{1,i} = r_i\,\sigma_{\mathrm{DR2},i}/\sigma_{\mathrm{DR1},i}$, which lies between $0.25$ and $0.49$ across the matched quantities (median $0.37$). The rule fixes only each quantity's correlation with itself, so pairs of quantities in the same bin need a further choice. Three ways of making that choice give a combined statistic between $28.9$ and $33.7$ ($p \leq 0.035$), and the strictly nested rule $\alpha_{1,i} = C_{\mathrm{DR2},ii}/C_{\mathrm{DR1},ii}$ gives $26.0$.

The second variant lets each bin take its own fraction, anywhere between zero and the smallest nested value $C_{\mathrm{DR2},ii}/C_{\mathrm{DR1},ii}$ among its quantities, so that no cross-release correlation exceeds that of strictly nested data. Its worst case is $16.2$ ($p \leq 0.0619$), which falls below the threshold. The one-fraction minimum of $64.7$ is therefore not a conservative bound for DESI's releases.

With no restriction on the cross-covariance the statistic has no bound at all, because the joint Gaussian laws consistent with the published marginals can drive it to $0$ or to arbitrarily large values. To see this, take a general cross-covariance $\Sigma_{12} = \mathrm{Cov}(X_1, X_2)$ of the matched quantities. The construction of \eqref{eq:Mjoint} carries over with $y = X_2 - \Sigma_{12}^T C_{\mathrm{DR1}}^{-1} X_1$, mean $\mu_2 - \Sigma_{12}^T C_{\mathrm{DR1}}^{-1} \mu_1$ and covariance $C_{\mathrm{DR2}} - \Sigma_{12}^T C_{\mathrm{DR1}}^{-1} \Sigma_{12}$, and the unmatched quasar term stays independent as above. A cross-covariance $\Sigma_{12}$ is admissible when the joint covariance is positive semi-definite, that is, when the singular values of $C_{\mathrm{DR1}}^{-1/2} \Sigma_{12}\, C_{\mathrm{DR2}}^{-1/2}$, the correlations between the two releases direction by direction, are all at most one. The laws that drive the statistic to its extremes have a rank-one cross-covariance inside a single bin, along a direction in which DR2 carries no new information. Suppose the innovation along that direction is exactly what the null predicts while its variance vanishes. Then every grid point's log-likelihood ratio diverges to $-\infty$ and the statistic goes to $0$. If the innovation differs from the null prediction instead, the ratio diverges to $+\infty$ for the grid points it favours. Requiring each quantity's DR1 and DR2 measurements to be non-negatively correlated does not remove these laws.

A bound returns only once we cap the correlation between the releases in every direction, and the cap has to be small. With the cap at $\rho$, the worst case over the full class is $2.3$ at $\rho = 0.3$, $0.09$ at $0.5$, $0.005$ at $0.7$ and $0.003$ at $0.75$. The class block-diagonal by bin gives $6.2$, $0.35$, $0.016$ and $0.007$ at the same caps. Of these, $0.75$ is the first tabulated cap that contains the $\alpha_1 = 1/3$ member, whose largest such correlation is $0.709$. The worst case falls below $20$ once the cap exceeds about $0.18$ ($0.22$ block-diagonal) and below $1$ once it exceeds $0.35$ ($0.43$). Finer families collapse faster, from $64.7$ with one fraction to $16.2$ with one per bin, $1.09$ with one per matched quantity and $0$ with none imposed. A model-based combination therefore needs a bound on how strongly the releases may be correlated in every direction, whereas the equal-weight average needs none.

The public products allow an internal-consistency test of the family, and the test excludes no member. If the family holds, the difference of the two releases' residuals about their own predictions, $d = (X_2 - \mu_2) - (X_1 - \mu_1)$, has mean zero under \LCDM\ and covariance $(1 - 2\alpha_1)\, C_{\mathrm{DR1}} + C_{\mathrm{DR2}}$ (Figure~\ref{fig:checks}b; subtracting the predictions corrects for the two bins whose effective redshifts differ between the releases, by at most $0.44\,\sigma_{\mathrm{DR2}}$). At $\alpha_1 = 1/3$ the test gives $\chi^2 = 11.40$ for 11 degrees of freedom ($p = 0.41$). Independent releases would give $5.80$ ($p = 0.89$), and even at the edge of the admissible range, $\alpha_1 = 0.47$, the test gives $19.0$ ($p = 0.061$). Stronger shared noise, such as $\mathrm{Cov}(X_1, X_2) = 0.6\, C_{\mathrm{DR1}}$, is excluded a priori rather than by the data, because no joint law with the published marginals has it. The largest pull is the LRG2 $D_M$ difference, $+1.96$ (Section~\ref{sec:loo}).

The test is also weak. If every matched DR2 measurement is offset by the same fraction of its DR2 $\sigma$, the test detects that offset at the $5\%$ level with probability one half only once it reaches $1.02\sigma$, about sixty times the calibration tolerance below.

\begin{table}[t]
\centering
\small
\begin{tabular}{@{}lccc@{}}
\toprule
Offset $b$ ($\sigma_{\mathrm{DR2}}$) & $\alpha_1 = 1/3$ & Worst case ($\alpha_1$) & DR2 alone \\
\midrule
0 & 66.7 & 64.7 (0.31) & 33.97 \\
0.01 & 43.6 & 43.1 (0.32) & 22.3 \\
0.02 & 9.77 & 8.85 (0.38) & 5.04 \\
0.05 & 0.0070 & $4.4 \times 10^{-5}$ (0.46) & 0.0037 \\
\midrule
Tolerance $b^*$ & 0.0160 & 0.0159 (0.35) & 0.0111 \\
\bottomrule
\end{tabular}
\caption{Calibration budget. A coherent offset $b$ shifts every matched DR2 measurement by $b\,\sigma_{\mathrm{DR2}}$ in the same direction, relative to DR1 for the combined statistics and relative to the \LCDM\ prediction for DR2 alone. Each entry is the statistic divided by its largest mean under \LCDM\ with that offset, so it remains an e-value when an offset of this pattern, of either sign and of size at most $b$, may be present. The worst case is the minimum of the one-fraction family over $0 \leq \alpha_1 < 0.4701$, reached at the $\alpha_1$ in parentheses; $b^*$ is the offset at which each value falls to the threshold $20$.}
\label{tab:budget}
\end{table}

A relative-calibration offset between the two pipelines of $0.016\,\sigma_{\mathrm{DR2}}$ is enough to remove the model-based rejection, and that offset lies far below anything the null test can see (Table~\ref{tab:budget}). The family assumes that each release measures its own prediction without offset, and the table relaxes that assumption. Each entry divides the statistic by its largest possible mean under \LCDM\ when every matched DR2 measurement is offset relative to DR1 by $b$ times its DR2 uncertainty in the same direction (either sign; the Gaussian moment-generating function also covers every smaller offset of the same pattern), so the entry remains an e-value when such an offset may be present. That ratio falls to $20$ at $b^* = 0.0160$ for $\alpha_1 = 1/3$ and at $0.0159$ for the one-fraction worst case. The tolerance is a sufficient condition for the bound to survive, though not a necessary one.

The tolerance measures the sensitivity of a fixed-prior mixture e-value to a mean offset, and it is not specific to the cross-release model. The same budget applied to DR2 alone, for an offset relative to the \LCDM\ prediction, gives $0.0111\sigma$. On the Narrow box of the same grid the $\alpha_1 = 1/3$ tolerance is $0.062\sigma$, and on the Default box the largest mean comes from the far corner, $(w_0, w_a) = (-0.5, 1.0)$. Offsets on the DR1 side are far less damaging, because DR1 enters the combined statistic twice with opposite signs. The DR1-side tolerance of the $\alpha_1 = 1/3$ statistic is $1.66\sigma_{\mathrm{DR1}}$, a hundred times the DR2-side value in units of each release's own uncertainty.

Public information cannot establish agreement between the two pipelines at the $0.016\sigma$ level. DESI has not published the relative DR1--DR2 offset, although the collaboration could estimate it from mocks run through both pipelines. For scale, DESI's post-unblinding pipeline updates within DR1 moved BAO measurements by up to $0.6\sigma$~\citep{DESI2024}.

The budget covers the 11 matched quantities only. Offsets on the unmatched QSO measurements therefore lie outside it, and so do common-mode offsets that shift both releases alike, which cancel in $d$.

Simulations under the $\alpha_1 = 1/3$ model check these calculations but do not validate the model. With no offset the simulated false-positive rate of that statistic at threshold $20$ is about $0.03\%$ ($10^5$ draws), and the equal-weight average of Section~\ref{sec:eprocess} has null mean exactly one in every simulated world, including one outside the family.

\subsection{Planck-fixed background cosmology justification}
\label{app:fixed-planck}

We fix $\{h, \Omega_m, r_d\}$ to their Planck 2018 values rather than marginalising over them, so the e-values test \LCDM\ and the chosen background jointly. BAO-only analyses commonly make this simplification, because BAO summary statistics depend on $\{h, \Omega_m, r_d\}$ and on $\{w_0, w_a\}$, and the CMB constrains the first set independently of dark-energy dynamics (only $\Omega_b h^2$ of the remaining \LCDM{} parameters enters, through $r_d$, while $n_s$, $A_s$ and $\tau$ do not enter the BAO observables). The $w_0w_a$CDM alternative therefore adds $k = 2$ parameters over \LCDM. BAO alone cannot constrain $\{h, \Omega_m, r_d\}$ jointly with $\{w_0, w_a\}$ without informative CMB priors, and marginalising over them would need Boltzmann-code computations beyond the one-dimensional distance integrals we use here. The cost is an implicit assumption that Planck $h \approx 0.677$ is correct.

The e-values depend strongly on which Planck column we use, because a coherent shift in $(h, \Omega_m, r_d)$ moves all 13 predicted $D/r_d$ ratios together. The triple we use in the body mixes columns, so this dependence matters here. Its $(h, \Omega_m) = (0.6766, 0.3111)$ matches the +lensing+BAO chain, whose self-consistent sound horizon is $r_d = 147.21 \pm 0.23$~Mpc, whereas its $r_d = 147.05$~Mpc matches the TT,TE,EE+lowE column (the +lensing column has $147.09$). We therefore rerun the full pipeline for the mixed triple and for three self-consistent Planck columns (Table~\ref{tab:background-sensitivity}).

\begin{table}[t]
\centering
\small\setlength{\tabcolsep}{4pt}
\begin{tabular}{@{}lcccc@{}}
\toprule
Background & $M_{\mathrm{DR1}}$ & $M_{\mathrm{DR2}}$ & $M_{\mathrm{DR2}}^{\mathrm{no\,LRG2}}$ & Reject \LCDM \\
\midrule
(a) this work $(0.6766, 0.3111, 147.05)$ & 1.05 & 33.97 & 0.49 & $\checkmark$ \\
(b) TT,TE,EE+lowE+lensing $(0.6736, 0.3153, 147.09)$ & 2.89 & 381.7 & 2.37 & $\checkmark$ \\
(c) +lensing+BAO $(0.6766, 0.3111, 147.21)$ & 0.65 & 8.70 & 0.22 &  \\
(d) TT,TE,EE+lowE $(0.6727, 0.3166, 147.05)$ & 4.80 & 1417 & 5.55 & $\checkmark$ \\
\bottomrule
\end{tabular}
\caption{Sensitivity of the per-release mixture to the fixed background $(h, \Omega_m, r_d/\mathrm{Mpc})$, Default prior throughout. Row (a) is the baseline used in the body; rows (b)--(d) are self-consistent Planck 2018 columns. A tick marks $M_{\mathrm{DR2}} \geq 20$.}
\label{tab:background-sensitivity}
\end{table}

The per-release mixture spans a factor of $163$ across the four backgrounds, and the known CMB--BAO tension in $\Omega_m$ drives the spread. The mixture reaches only $8.70$ under the self-consistent +lensing+BAO column (c) and does not reject, but it rejects far more strongly than the baseline under the CMB-only columns (b) and (d). DESI DR2 BAO alone prefers $\Omega_m \approx 0.297$ under \LCDM, so a higher-$\Omega_m$ CMB-only background leaves a coherent residual across the 13 measurements, and the $(w_0, w_a)$ family is well placed to absorb that residual. The +BAO column has instead already absorbed part of the tension into the background itself, at the price of calibrating the null on BAO data of the kind under test. A fixed-background e-value therefore measures the background tension in part, rather than $(w_0, w_a)$ evolution alone.

Three results hold in every column. The fixed-background MLE stays in the $w_0 > -1$, $w_a < 0$ quadrant. DR1 alone never rejects ($M_{\mathrm{DR1}} \leq 4.8$). No column rejects once we remove LRG2, though under the CMB-only backgrounds the values without LRG2 ($2.37$ and $5.55$) stay above the baseline's $0.49$.

We quote column (a) throughout the body, because we fixed it as the baseline before computing any e-value, and we treat the background as part of the null rather than as a nuisance parameter. The spread across the columns is wide enough that readers should treat the headline value as specification-dependent in this respect too, alongside its dependence on the prior (Section~\ref{sec:prior-sens}) and on LRG2 (Section~\ref{sec:loo}). The worst case of Proposition~\ref{prop:family}, applied to the choice among these four columns, gives $8.70$ ($p \leq 0.115$), so no rejection from DR2 alone survives it.

\subsection{Validity and calibration of the leave-one-out decomposition}
\label{app:loo-decomp}

Not every aggregate one could build from Table~\ref{tab:loo} carries a guarantee. The leave-one-out averages, the look-elsewhere-corrected mean and the per-bin-independent product do, but the product of the leave-one-out e-values does not. Those LOO folds share training data, so we do not report their product (note to Proposition~\ref{prop:loo}). By contrast, the per-bin-independent product $\prod_k M_k^{\mathrm{REGROW},\delta=2}$ is valid, because its per-bin mixtures use a fixed, data-independent prior and because DESI's published covariance is block-diagonal across bins.

The same constructions under the Default prior give $\prod_k M_k = 1.8 \times 10^{-3}$ and $\bar M = 3.5$, with only LRG2 individually above one ($M_{\mathrm{LRG2}} = 22.8$). The wide per-bin prior pays the Bartlett penalty of Section~\ref{sec:prior-sens} in every bin separately, which explains these values.

Reading the LRG2-dominated decomposition as a localised feature assumes that a smooth $w(z)$ would spread its evidence across the bins, so we test that assumption with a parametric bootstrap. We generate synthetic DR2 vectors under two truths, $w_0w_a$CDM at the DR2 BAO MLE and \LCDM. Each vector consists of the 13 predicted observables at the fixed background of Section~\ref{sec:methods} plus Gaussian noise drawn with the published $13\times13$ covariance. We then pass each vector through exactly the leave-one-out pipeline of Section~\ref{sec:loo} and record each bin's share of $\sum_k E_k^{\mathrm{LOO}}$ ($500$ and $200$ realisations respectively, fixed seeds, code release).

\begin{table}[t]
\centering
\small
\begin{tabular}{@{}lcc@{}}
\toprule
 & $w_0w_a$CDM truth (MLE) & \LCDM\ truth \\
\midrule
Median max-bin LOO share & $55\%$ & $22\%$ \\
Some bin carries $\geq 75\%$ & $23\%$ of realisations & $0/200$ \\
Max-bin share $\geq$ observed $78.6\%$ & $18\%$ & $0/200$ \\
LRG2 specifically $\geq 78.6\%$ & $3.2\%$ & $0/200$ \\
Most frequent dominant bin & LRG3$+$ELG1 ($37\%$) & -- \\
\bottomrule
\end{tabular}
\caption{Expected concentration of the leave-one-out evidence, from the parametric bootstrap of this appendix ($500$ $w_0w_a$CDM and $200$ \LCDM\ realisations). In the data, LRG2 carries $78.6\%$ of $\sum_k E_k^{\mathrm{LOO}}$.}
\label{tab:loo-benchmark}
\end{table}

Under \LCDM\ truth the evidence never concentrates as strongly as in the data. The median max-bin share is $22\%$, and no realisation reached $75\%$ (Table~\ref{tab:loo-benchmark}), so the observed concentration is a real departure from the null. A smooth signal, by contrast, concentrates often, because the bins differ widely in redshift lever arm and precision. One bin holds at least $75\%$ of the sum in $23\%$ of realisations, and the observed $78.6\%$ carried by LRG2 lies at the 82nd percentile of that distribution.

Concentration alone therefore does not separate a localised feature from a smooth signal, and neither does the identity of the dominant bin, but the share carried by LRG2 separates them in part. Smooth signals make LRG2 the dominant bin in $27\%$ of realisations (LRG3$+$ELG1 in $37\%$), yet give it at least the observed $78.6\%$ in only $3.2\%$. That figure conditions on LRG2 after we observed its dominance, whereas the unconditioned rate, at which \emph{some} bin reaches the observed share, is $18\%$. The benchmark distribution already accounts for the bins' differing sensitivities, because every synthetic vector carries the data's bin structure (same number of measurements per bin, published covariance). The bootstrap does not, however, model bin-specific systematics beyond the published Gaussian covariance, so a systematic of that kind in LRG2 remains an open interpretation.

The decomposition is exact at the bin level because we build it from e-values. A single bin supplies one or two measurements, too few for Wilks asymptotics. Each held-out likelihood ratio, by contrast, is an exactly valid finite-sample e-value (Proposition~\ref{prop:loo}), so we can average the held-out ratios, and multiply the per-bin mixtures across independent bins, with the Type-I guarantee intact.

\subsection{Minimum prior concentration \texorpdfstring{$\delta^*$}{delta*} for the Fisher-Gaussian family}
\label{app:delta-star}

The data-supported window of Section~\ref{sec:prior-sens} follows from a closed form for the Fisher-shaped Gaussian shell priors $\pi_\delta = \mathcal N(\theta_0,\, \delta^2 F^{-1})$ centred at \LCDM, and the shell rejects for widths between about $0.87$ and $13.5$ Fisher-$\sigma$. We derive it from the local Gaussian (Wilks/Laplace) approximation with curvature $F$ at $\theta_0$ and deviation $\Delta\chi^2 = (\hat\theta_{\mathrm{MLE}} - \theta_0)^T F (\hat\theta_{\mathrm{MLE}} - \theta_0)$, under which the mixture e-value is
\begin{equation}
M(\delta) = \frac{\exp\!\left(\tfrac{1}{2}\,\Delta\chi^2\,\dfrac{\delta^2}{1+\delta^2}\right)}{(1+\delta^2)^{k/2}}.
\label{eq:M-delta}
\end{equation}
For DESI DR2 ($\Delta\chi^2 \approx 16.5$ here, versus $16.86$ for the exact likelihood; $k = 2$), the equation $M(\delta) = 20$ has two roots. The lower root $\delta^*_{\mathrm{lower}} \approx 0.87$ Fisher-$\sigma$ is the narrowest aligned prior that reaches the threshold. At the upper root $\delta^*_{\mathrm{upper}} \approx 13.5$, the Occam factor $(1+\delta^2)^{-k/2}$ drives $M$ back below $20$. Between them, the maximising width $\delta_{\mathrm{max}} = \sqrt{\Delta\chi^2/k - 1} \approx 2.7$ gives $M_{\mathrm{max}} \approx 170$.

Only the lower part of that window is physically meaningful. The DR2 MLE sits at $\delta_{\mathrm{MLE}} = 4.06$, so the upper root lies far beyond the region the data prefer, and the usable shell range is $\delta \in [0.87, 4]$. This window belongs to the Gaussian-shell family alone, and it does not transfer to the REGROW family of Table~\ref{tab:summary}.

\subsection{Limitations}
\label{app:limitations}

Fixing $\{h, \Omega_m, r_d\}$ to Planck 2018 values (Appendix~\ref{app:fixed-planck}) suppresses the $w_a$--$\Omega_m$ and $w_0$--$\Omega_m$ degeneracies, so the e-values we report probe the dark-energy sector at a fixed background rather than over the full parameter space. Marginalising would pull in two directions at once. A broadened alternative dilutes the mixture, while a fit with the degeneracies opened can reabsorb residual structure into $(\Omega_m, h)$ and raise $\Delta\chi^2$. The net sign is therefore not guaranteed, though we expect the fixed-background e-values to be inflated relative to the marginalised case.

The concern is sharper for LRG2. Its $D_H/r_d$ residual at $z = 0.706$ is sensitive to $h$ and $\Omega_m$ through the expansion rate, so a free background could reattribute part of the anomaly from $w(z)$ to a shifted $(\Omega_m, h)$. We plan the marginalised analysis that would separate the two, but we have not carried it out here.

The $z=1$ data-split is underpowered for $w_a$ on the $z<1$ training set, where the CPL factor $(1-a)$ is small. We fit $(w_0, w_a)$ there and score the held-out $z \geq 1$ set, which gives $E_{\mathrm{split}} = 1.43$, and the median is $\approx 1.2$ even under $w_0w_a$CDM truth. Under an alternating-bin split the result also swings by a factor of about $66$, depending on whether LRG2 falls in the training set or the scored set. We therefore do not rely on the single split. Our within-DESI inferences rest instead on the per-release mixture (Section~\ref{sec:results}) together with the LOO and per-bin-independent decompositions (Section~\ref{sec:loo}), all of which use all the data.

All of these results condition on the CPL form $w(a) = w_0 + w_a(1-a)$, and conclusions drawn from DESI data are known to be sensitive to that choice~\cite{Nesseris2025, Lee2026, Alam2026}. A small SN--BAO distance-duality mismatch can likewise shift the CPL inference away from \LCDM~\cite{AfrozMukherjee2026}. Extending the e-value framework to Pad{\'e}, principal-component and free-functional parametrisations (see Ref.~\cite{Nesseris:2013bia} for an overview) would test how much of the signal survives a less restrictive description of $w(z)$.

Finally, readers should not over-read the cross-dataset cross-prediction of Appendix~\ref{app:cross}. It is a pre-specified-alternative e-value, valid by Proposition~\ref{prop:lr}, but it does not combine evidence across the SN, CMB and BAO release histories. The natural combination would multiply $M^{\mathrm{DESI}}_t \cdot M^{\mathrm{SN}}_t \cdot M^{\mathrm{CMB}}_t$ across independent surveys, combine each survey's own releases by a rule fixed in advance (Section~\ref{sec:eprocess}), and handle the shared Planck-fixed background consistently in the constituent likelihoods. We leave this construction for follow-up, because it requires the full SN and CMB likelihoods rather than the published joint-posterior summaries we use here.

\subsection{Joint BAO+CMB e-values: methodological details}
\label{app:joint-cmb}

We compare with the DESI BAO + Planck CMB combination of Ong et al.~\citep{OngBayesian2025} by adding the compressed Planck 2018 likelihood of Chen, Huang and Wang~\citep{Chen2019Compressed}, with $r_s(z_*) = 144.65$~Mpc fixed (code release; distinct from the BAO drag-epoch $r_d = 147.05$~Mpc). This compressed statistic approximates the full Planck likelihood of Ong et al.~\citep{OngBayesian2025}, so we do not reproduce their absolute Bayes-factor magnitude. The joint MLE is $(w_0, w_a) = (-0.818, -0.695)$ with $\Delta\chi^2 = 14.78$ ($\sigma = 3.42$ Wilks), which matches DESI's $\sim 3.1\sigma$ to within the fixed-versus-marginalised early-universe difference. The e-value summaries are
\begin{equation}
\bar{E}^{\mathrm{LOO}}_{\mathrm{BAO+CMB}} = 7.07, \qquad
E_{\mathrm{mix,Default}}^{\mathrm{BAO+CMB}} = 2.19.
\label{eq:joint-evalues}
\end{equation}

The mixture e-value reproduces the prior-width reduction that Ong et al.\ found. Under the Narrow, Default, Wide and Ong et al.\ boxes of Table~\ref{tab:summary}, the joint BAO$+$CMB mixture e-value is $7.2$, $2.19$, $0.91$ and $0.091$, so no prior reaches the rejection threshold, and the value falls monotonically with prior width. At their wide priors, the e-value points the same way as their $\ln\mathcal{B} = -0.57$. The e-value makes this prior-width dependence explicit, as a continuous family rather than a single number.

\subsection{Cross-dataset cross-prediction}
\label{app:cross}
Each SN+CMB joint MAP is a pre-specified $(w_0, w_a)$ point alternative against DESI BAO, so its likelihood ratio is a valid e-value (Proposition~\ref{prop:lr}; see also~\citep[Sec.~1]{GrunwaldSafeTesting2024}). The e-value averaged over each compilation's full posterior is valid too, by linearity of expectation. The point e-values against DR2 BAO, evaluated at the published MAPs, are $9.6$ for Pantheon+$\,$+$\,$CMB at $(w_0, w_a) = (-0.851, -0.70)$~\citep{BroutPantheon2022}, $47$ for DES-Y5$\,$+$\,$CMB at $(-0.730, -1.17)$~\citep{DESY5_2024} and $44$ for Union3$\,$+$\,$CMB at $(-0.699, -1.05)$~\citep{RubinUnion3_2023}. The posterior-mean e-values are $342$, $240$ and $128$. The posterior-mean values are the more stable measure, because $E$ is exponential in $\Delta\chi^2$ and a $\sim 1\sigma$ MAP shift therefore swings the point $E$ by a large factor.

The three compilations' posterior-mean e-values score well above the threshold against DESI BAO in $w_0w_a$CDM, which is consistent with~\citep{OngBayesian2025, WangMota2025, ColgainEvolved2025}, though their MAPs differ at $\sim 1\sigma$ in $(w_0, w_a)$. The large values are expected, and they do not amount to independent confirmation. SN$+$CMB data that prefer $w_0w_a$CDM on their own inform these alternatives, so the alternatives sit near the DESI MLE, where the likelihood ratio is large.

\acknowledgments

We thank Peter Gr\"unwald (CWI Amsterdam) for detailed feedback on the e-value methodology and Ruodu Wang (University of Waterloo) for confirming the validity of the leave-one-out averaging argument. We thank Eoin~{\'O}~Colg{\'a}in, Ioannis~D.~Gialamas, Tian-Nuo~Li and Suvodip~Mukherjee for comments on the first arXiv version. DFM and AT thank the Research Council of Norway for their support and the resources provided by UNINETT Sigma2, the National Infrastructure for High-Performance Computing and Data Storage in Norway. AT is also supported by the European Union's Horizon Europe research and innovation programme under the Marie Sk\l odowska-Curie grant agreement No. 101126636.
We also thank the DESI collaboration for the public release of the DR1 and DR2 BAO data products used in this analysis. The authors used Claude (Anthropic; Opus 5.5 and Fable 5.1) for drafting, analysis code and literature searches, checked all results and text themselves, and take full responsibility for the content.

\section*{Data and code availability}

{\sloppy Analysis code and data files, including the scripts for the cross-release analyses of Section~\ref{sec:eprocess} and Appendix~\ref{app:hierarchical}, are available at \url{https://github.com/jinyoungkim927/desi-evalue-paper} (release v3). DESI BAO data were taken from \url{https://github.com/CobayaSampler/bao_data}. The statistics use only the public DESI DR1 and DR2 BAO mean vectors and covariance matrices distributed there. The posterior summaries of Table~IV of Ref.~\citep{DESI2025} and the correlations of DESI's DR1--DR2 consistency check enter only the comparisons of Sections~\ref{sec:methods} and~\ref{sec:eprocess} and Appendix~\ref{app:hierarchical}.\par}
\bibliographystyle{JHEP}
\bibliography{refs}

\end{document}